\documentclass{article}
\usepackage{graphicx}
\usepackage{mathtools}
\usepackage{authblk}
\usepackage{setspace}
\usepackage{hyperref}
\hypersetup{backref, colorlinks=true, linkcolor=blue, citecolor=blue, urlcolor = blue}

\usepackage{amsmath}
\usepackage{amssymb}
\usepackage{amsthm}
\usepackage{enumitem}

\usepackage{mathtools}

\usepackage{epsfig}
\usepackage{psfrag}
\usepackage{cite}
\usepackage{cuted}
\usepackage[mathscr]{euscript}
\usepackage{microtype} %improves the spacing between words and letters
\usepackage{ushort}
\usepackage{environ}
\usepackage{tabu}
\usepackage{capt-of}
\usepackage[ruled,vlined]{algorithm2e}
\usepackage{stmaryrd}
\usepackage{dsfont}

\usepackage{comment}
\usepackage[framemethod=tikz]{mdframed}
\mdfsetup{skipabove=0pt,skipbelow=0pt}

\usepackage{tikz}
\usepackage{mathtools}

\ifx\eqref\undefined
	\newcommand{\eqref}[1]{~(\ref{#1})}
\fi
\ifx\mod\undefined
	\def\mod{\mathop{\rm mod}}
\fi

\newcommand{\bigo}[1]{\mathcal{O}\left(#1\right)}

\newcommand{\beq}{\begin{equation}}
\newcommand{\eeq}{\end{equation}}

\def\exp{\mathop{\rm exp}}

\def\EE{\mathbb{E}\,}

\def\diag{\mathop{\rm diag}}

\newcommand{\calD}{\mathcal{D}}
\newcommand{\calE}{\mathcal{E}}
\newcommand{\calF}{\mathcal{F}}

\newcommand{\calH}{\mathcal{H}}

\newcommand{\calI}{\mathcal{I}}
\newcommand{\calJ}{\mathcal{J}}

\newcommand{\calL}{\mathcal{L}}

\newcommand{\calP}{\mathcal{P}}

\newcommand{\calT}{\mathcal{T}}
\newcommand{\calU}{\mathcal{U}}

\def\sA{\mathsf{A}}
\def\sB{\mathsf{B}}

\def\unifto{\mathop{{\mskip 3mu plus 2mu minus 1mu%
	\setbox0=\hbox{$\mathchar"3221$}%
	\raise.6ex\copy0\kern-\wd0%
	\lower0.5ex\hbox{$\mathchar"3221$}}\mskip 3mu plus 2mu minus 1mu}}

\ifx\lesssim\undefined
\def\simleq{{{\mskip 3mu plus 2mu minus 1mu%
	\setbox0=\hbox{$\mathchar"013C$}%
	\raise.2ex\copy0\kern-\wd0%
	\lower0.9ex\hbox{$\mathchar"0218$}}\mskip 3mu plus 2mu minus 1mu}}
\else
\def\simleq{\lesssim}
\fi

\ifx\gtrsim\undefined
\def\simgeq{{{\mskip 3mu plus 2mu minus 1mu%
	\setbox0=\hbox{$\mathchar"013E$}%
	\raise.2ex\copy0\kern-\wd0%
	\lower0.9ex\hbox{$\mathchar"0218$}}\mskip 3mu plus 2mu minus 1mu}}
\else
\def\simgeq{\gtrsim}
\fi

\newcounter{keepeqno}

\newtheorem{definition}{Definition}[section]

\newtheorem{lemma}[definition]{Lemma}

\newtheorem{theorem}[definition]{Theorem}
\newtheorem{corollary}[definition]{Corollary}

\usepackage[USenglish]{babel}

\theoremstyle{remark}

\theoremstyle{definition}

\newcommand{\R}{\mathbb{R}}

\newcommand{\C}{\mathbb{C}}

\newcommand{\Tr}{\mathrm{Tr}}

\renewcommand{\Re}{\operatorname{Re}}

\newcommand{\supp}{\text{supp}}

\newcommand{\ket}[1]{\left\vert #1 \right\rangle}

\newcommand{\bra}[1]{\left\langle #1 \right\vert}

\usepackage[a4paper, margin=1in]{geometry}

\title{Optimal detection of dissipation in Lindbladian dynamics}
\author{Yiyi Cai}
\affil{University of Cambridge}
\date{}
\begin{document}

\maketitle
\begingroup
\renewcommand{\thefootnote}{}
\footnotetext{Corresponding email: yiyicai0615@gmail.com}
\endgroup

\begin{abstract}
    Experimental implementations of Hamiltonian dynamics are often affected by dissipative noise arising from interactions with the environment.
    This raises the question of whether one can detect the presence or absence of such dissipation using only access to the observed time evolution of the system.
    We consider the following decision problem: given black-box access to the time-evolution channels $e^{t\calL}$ generated by an unknown time-independent Lindbladian $\calL$, determine whether the dynamics are purely Hamiltonian or contain dissipation of magnitude at least $\epsilon$ in normalized Frobenius norm.
    We give a randomized procedure that solves this task using total evolution time $\bigo{\epsilon^{-1}}$, which is information-theoretically optimal. 
    This guarantee holds under the assumptions that the Lindblad generator has bounded strength and its dissipative part is of constant locality with bounded degree.
    Our work provides a practical method for detecting dissipative noise in experimentally implemented quantum dynamics.
\end{abstract}

% \newpage
\tableofcontents
\newpage

\section{Introduction}
Quantum systems are rarely perfectly isolated.
In realistic experimental platforms, the evolution of a quantum device is typically governed not only by a coherent Hamiltonian but also by unwanted interactions with the surrounding environment.
Such open-system dynamics are well described by Lindblad generators \cite{lindblad1976generators, gorini1976completely}, which describe the most general Markovian evolution of a quantum system.
In particular, dissipative effects arising from environmental coupling, imperfect control, or thermal noise can significantly alter the intended evolution of a quantum device and degrade stored quantum information.
Controlling and suppressing such dissipation is one of the central challenges in the development of reliable quantum technologies, especially in the regime of early fault-tolerant hardware.
Detecting and characterizing such dissipative processes is therefore a fundamental task in the validation and verification of quantum devices.
This is particularly important in the context of error mitigation and quantum error correction, where one aims to suppress or eliminate dissipative noise at the logical level, and requires reliable diagnostics to assess whether such suppression has been successful.

More broadly, the problem of certifying whether a quantum object behaves as intended has been extensively studied across different settings, including quantum states \cite{buadescu2019quantum, huang2025certifying}, unitary evolutions \cite{jeon2025query}, and quantum channels \cite{rosenthal2024quantum, fawzi2023quantum}.
For instance, quantum state certification aims to distinguish whether an unknown state matches a target state or is far from it using significantly fewer samples than full tomography, while analogous tasks have been considered for unitary dynamics and quantum channels.

While these works provide powerful certification guarantees, they typically rely on strong access models, such as the ability to prepare multiple copies or perform global measurements.
In many situations, however, the microscopic description of the dynamics is not directly accessible. 
Instead, the system can only be probed through black-box access to its time-evolution channel.
This raises a natural question:
\begin{center}
    \textit{Can we efficiently detect the presence of dissipation using only access to the system dynamics?}
\end{center}
More concretely, we consider the following decision problem: given oracle access to the evolution channel $\calE_t = e^{t\calL}$ generated by an unknown Lindbladian $\calL$ for a chosen time $t$, can we distinguish the case where the dynamics is purely Hamiltonian from the case where a dissipative component of non-negligible strength is present?

At first glance this problem appears challenging. 
Dissipative effects are often subtle at short times and may only become visible after the system has evolved for sufficiently long durations.
Moreover, directly learning the full Lindblad generator is difficult in large systems, since the generator acts on an operator space whose dimension grows exponentially with the number of qubits.
Recent works have explored algorithms for learning
open-system dynamics under various structural assumptions \cite{ivashkov2026ansatz, francca2025learning, stilck2024efficient}; however, such approaches typically require substantial experimental resources and detailed access to the system.
While learning the full generator would in principle reveal whether dissipation is present, such reconstruction problems are typically far more demanding than certification and detection tasks that only aim to detect a specific property of the dynamics.
These considerations suggest that a practical approach should instead rely on carefully chosen observables that reveal signatures of dissipation without requiring full process tomography. 

In this work we show that such an efficient detection procedure is possible.
At a high level, our approach exploits the fact that dissipation produces irreversible contraction, whereas purely Hamiltonian evolution produces coherent rotations.
Bell sampling of the evolution $e^{t\calL}$ is sensitive to both effects, since it measures how strongly the implemented channel resembles the identity channel.
We therefore first apply a randomized Pauli-twirling reduction which removes the Hamiltonian contribution at the generator level.
For the resulting generator-twirled dynamics, purely Hamiltonian evolution becomes the identity evolution, while dissipation produces decay in a Bell-sampling statistic.
Detecting this decay allows us to infer the presence of dissipation without learning the full generator.
To implement this generator-level twirling operator using only black-box access to finite-time channels $e^{t\calL}$, we perform trotterization and show that dissipation of strength $\epsilon$ can be detected using total evolution time $\bigo{1/\epsilon}$, which is optimal up to constant factors.
% At a high level, our approach exploits the fact that dissipative
% dynamics gradually suppress certain measurable signals in the system, whereas purely Hamiltonian evolution preserves them.
% We probe this effect using global observables obtained via Bell sampling, which quantify how strongly the implemented evolution resembles the identity channel.
% Under purely Hamiltonian dynamics this quantity remains constant over time, while the presence of dissipation causes it to decay.
% Detecting such a decay allows us to infer the presence of dissipation without learning the full generator.
% The key challenge is to perform this detection procedure using as little total evolution time as possible.
% We show that dissipation of strength $\epsilon$ can be detected using total evolution time $\bigo{1/\epsilon}$, which is optimal up to constant factors.
Consequently, the resulting procedure scales efficiently with the relevant system parameters while achieving optimal dependence on the evolution time.
Beyond the detection task studied here, the reduction used here may also be useful as a primitive for learning open-system dynamics: by separating coherent rotations from dissipative decay and converting the latter into Pauli-diagonal spectral data, it suggests a possible route toward learning dissipative components without reconstructing the full Lindblad generator.

\subsection{Main results}
We study the problem of certifying dissipation in Markovian open-system quantum dynamics given limited experimental access.
Let $\calL: \mathsf{L}(\C^d) \to \mathsf{L}(\C^d)$  be a time-independent Lindblad generator acting on an $n$-qubit system,
\begin{equation}
    \calL(\rho) = -i[H,\rho]  + \underset{\calD(\rho)}{\underbrace{\sum_{a=1}^m L_a\rho L_a^\dagger - \frac{1}{2}\{L_a^\dagger L_a, \rho\}}},
\end{equation}
where $H$ is an arbitrary Hamiltonian and $\calD$ is the dissipative part of the dynamics, generated by jump operators $\{L_a\}_{a}$.
The corresponding time evolution is given by the quantum channel $\calE_t := e^{t\calL}$ for $t \geq 0$.

We consider the following certification task. 
Given black-box access to $\calE_t$ at times of our choosing, one must distinguish, with high probability, between $\calD = 0$ (purely Hamiltonian dynamics) and $\|\calD\|_F \geq \epsilon$, where $\|\cdot\|_F$ denotes the normalized Frobenius norm on superoperators, as defined in Equation (\ref{eqn:norm_frob}) in Section \ref{subsect:superop}.
More specifically, the certification procedure must accept if $\calD = 0$ and reject if $\|\calD\|_F \geq \epsilon$.

Our main result shows that this certification task can be solved efficiently under standard assumptions that the Lindbladian has bounded strength and that its dissipative part satisfies locality constraints.

\begin{theorem}[Informal version of Theorem \ref{thm:final_certification}]\label{thm:informal_nfn}
    Consider the class of time-independent Lindblad generators $\calL = -i[H, \cdot] + \calD$ on $n$ qubits such that the jump operators defining $\calD$ have constant locality and bounded degree (see Definition  \ref{def:locality_and_degree}), and $\|\calL\|_{\diamond} \leq L$.
    Then there exists a randomized procedure which, given black-box access to the evolution channels $\{\calE_t:= e^{t\calL}\}_{t\geq 0}$, distinguishes between $\calD = 0$ and $\|\calD\|_F \geq \epsilon$ with high probability, using total time evolution $\bigo {\epsilon^{-1}}$ and query complexity $\bigo{ L^2/\epsilon^2}$. 
\end{theorem}
Theorem \ref{thm:informal_nfn} above is stated in terms of the normalized Frobenius norm on superoperators, which coincides exactly with their Pauli $2$-norm when expressed in in the Pauli basis.
More specifically, we define the the Pauli $p$-norm of an superoperator as taking the $\ell_p$ norm of their coefficients in the Pauli expansion.  
Indeed, we can represent the Lindbladian $\calL$ in the Pauli operator basis. 
Let $\{P_i\}_{i=1}^{d^2}$ denote the normalzed $n$-qubit Pauli operators, satisfying $\Tr(P_i P_j) = d \delta_{i,j}.$
Since $\calL$ is linear map on operators, it admits a matrix representation with coefficients 
\[L_{i,j} := \frac{1}{d} \Tr(P_i \calL(P_j)),\]
so that 
\[\calL(P_j) = \sum_i L_{i,j} P_i.\]
We define the Pauli $p$-norm of $\calL$ as the $l_p$-norm of the coefficient matrix $L = (L_{i,j})$.
We will in particular apply this representation to the dissipative part $\calD$, and denote its Pauli coefficients by $D_{i,j} := \frac{1}{d}\Tr(P_i \calD(P_j))$.
For $p \geq 2$, monotonicity of $\ell_p$ implies that a lower bound on the Pauli $p$-norm yields a lower bound on the Pauli $2$-norm, so the certification task is no harder than the normalized Frobenius setting ($p=2$).
For $1 \leq p < 2$, an analogous statement holds with an additional overhead depending on the number of nonzero Pauli coefficients of $\calD$, which arises from the norm conversion between Pauli $p$- and $2-$norms as well.
We thus introduce the following Corollary \ref{cor:p_norm} without a proof since no new algorithmic ideas are required beyond those of Theorem \ref{thm:informal_nfn}.

\begin{corollary}\label{cor:p_norm}
    Under the assumptions of Theorem \ref{thm:informal_nfn}, the same certification guarantees hold when the promise on the dissipative part $\calD$ is expressed in terms of the Pauli $p$-norm.
    In particular, for any $p \geq 2$, there exists a randomized procedure that distinguishes $\calD = 0$ from $\|\calD\|_p \geq \epsilon$ using total evolution time $\bigo{\epsilon^{-1}}$ with high success probability.
    For $1 \leq p < 2$, an analogous statement holds with total evolution time $\bigo{\|\calD\|_0^{1/p - 1/2}\epsilon^{-1}}$.
\end{corollary}

We remark that the $\epsilon$-dependence in our $\bigo{\epsilon^{-1}}$ bound is optimal.
Consider a single qubit undergoing evolution with the promise that either $\calD = 0$ (so $\calE_t = I$) or a depolarizing dissipator with rate $\gamma = \bigo{\epsilon}.$
In the latter case, the resulting evolution channel is 
\[\calE_t(\rho) = e^{-\gamma t} \rho + (1-e^{-\gamma t}) \frac{I}{2},\] so its deviation from the identity channel is at most $1 - e^{-\gamma t} \leq \gamma t$.
Thus, a constant distinguishing advantage requires $\gamma t = \Omega(1)$.
Since $\gamma = \bigo{\epsilon}$, this implies that the total evolution time satisfies 
$\Omega(1/\epsilon).$

\subsection{Technical overview}
We provide a high-level overview of the ideas underlying the detection procedure. 
Our goal is to detect whether a Lindbladian evolution contains a non-trivial dissipative component using only black-box access to the time-evolution channels  $\calE_t = e^{t\calL}.$
\paragraph{Bell sampling as a decay observable.}
The central observable in our detection procedure is obtained through an experimental primitive known as Bell sampling. 
Given black-box access to a quantum channel $\calE$, Bell sampling produces a distribution over Bell measurement outcomes. 
Specifically, it prepares a maximally entangled state $\ket{\Phi}$, applies the channel $\calE$ to one half, and measures in the Bell basis. 
We focus on the probability of obtaining the outcome $\ket{\Phi}$ after the channel evolution, which we denote by
\[I(\calE) := \Pr[\text{Bell outcome } \ket{\Phi}] = \frac{1}{d^2}\Tr(\calE).\]
This quantity admits a simple operational interpretation.
Viewing $\calE$ as a linear operator acting on matrices, $I(\calE)$ is the average overlap between an operator and its image under the channel.
Thus $I(\calE)$ measures how strongly the channel resembles the identity channel. 
% Equivalently, in the Pauli basis, $I(\calE)$ is the average of the diagonal entries of the channel's Pauli transfer matrix.
% We note that this raw statistic does not by itself distinguish dissipation from coherent Hamiltonian evolution. 
% If $\calE_t = e^{t\calL}$ with $\calL = -i [H, \cdot ]$, then $\calE_t$ is unitary conjugation.
% Although such evolutions preserves Hilbert-Schmidt norms, it can rotate operators away from themselves, and therefore $I(\calE_t)$ needs not remain equal to $1$.
% Indeed, 
% \[I(\calE_t) = \frac{|\Tr(e^{-itH})|^2}{d^2},\]
% which can be strictly smaller than $1$.
% Thus raw Bell sampling measures closeness to the identity channel, not dissipation alone.
% Unitary evolution acts as a rotation in operator space and preserves inner products, implying $I(\calE) = 1$ for all unitary channels.
% In contrast, dissipative dynamics introduce contraction, leading $I(\calE) < 1$.

When applied to Lindbladian evolution $\calE_t = e^{t\calL}$, this observable takes the form 
\[I(t) := I(\calE_t) = \frac{1}{d^2}\Tr(e^{t\calL}) = \frac{1}{d^2}\sum_{i}e^{t\lambda_i},\]
where $\{\lambda_i\}$ are the eigenvalues of $\calL$. 
Thus, the statistic $I(t)$ is the average of the factors $e^{t\lambda_i}$ that describe how different eigen-operators of $\calL$ evolve under the dynamics.
We formalize Bell sampling on quantum channels in Section \ref{subsec:bell_sampling}.
% Bell sampling provides a simple observable that captures how much of the identity component of a quantum channel remains after the evolution.
% Concretely, applying Bell sampling to a channel $\calE$ estimates the probability that the channel acts as the identity on one half of a maximally entangled state.
% This probability equals $I(\calE) = \frac{1}{d^2} \Tr(\calE)$, which measures the overlap of the channel with the identity superoperator.
% When applied to the Lindbladian evolution $\calE_t =e^{t\calL}$, this yields the statistic 
% \[I(t) := I(\calE_t) = \frac{1}{d^2}\Tr(e^{t\calL}) = \frac{1}{d^2}\sum_{i}e^{t\lambda_i},\]
% where $\{\lambda_i\}_{i}$ are the eigenvalues of $\calL$. 
% Thus, the statistic $I(t)$ is the average of the factors $e^{t\lambda_i}$ that describe how different components of operators evolve under the dynamics.
% We formalize Bell sampling on quantum channels in Section \ref{subsec:bell_sampling}.

However, this raw statistic does not by itself isolate dissipation. If the dynamics are purely Hamiltonian, then $\calL = -i[H, \cdot]$ generates a unitary conjugation channel. 
Such dynamics preserve Hilbert--Schmidt norms, but they can rotate operators away from themselves. 
Consequently, $I(e^{t\calL})$ needs not remain equal to $1$.
For this reason, our algorithm applies Bell sampling only after a randomized reduction called twirling that removes the coherent Hamiltonian contribution from the effective generator.
After this reduction, the purely Hamiltonian case corresponds to identity evolution, so the Bell identity probability remains equal to $1$.
In contrast, when dissipation is present, the reduced dynamics have negative decay rates, causing the Bell statistic to decrease over time.
Bell sampling of this twirled evolution then gives a decay statistic: if sufficiently many reduced decay rates are at least $\epsilon$, choosing $t = \bigo{\epsilon^{-1}}$ yields a detectable decrease with constant probability.
Our analysis in Section \ref{subsec:pauli_diag} first studies this decay primitive in the reduced, Pauli-diagonal setting, and in Section \ref{subsec:twirling} explains how the randomized reduction connects general Lindbladian dynamics to that setting.
% This quantity behaves differently depending on whether dissipation is present. 
% If the dynamics are purely Hamiltonian, then all eigenvalues $\lambda_i$ are purely imaginary, so $\calE_t = -i[H, \cdot]$ is a unitary conjugation channel that simply rotates operators without contracting them, and consequently $I(t) =1$ for all $t$.
% In contrast, the presence of dissipation introduces eigenvalues of $\calL$ with negative real part.
% This leads to $e^{t\lambda_i}$ factors to have magnitude strictly less than $1$, causing $I(t)$ to decrease as $t$ grows.
% Thus, a detectable drop in the Bell identity probability directly certifies the presence of dissipation.
% Our analysis in Section \ref{sec:decay_primitive} shows that if a constant fraction of eigenoperators $\{e^{t\lambda_{i}}\}$ decay at rate at least $\epsilon$, then choosing an evolution time $t = \bigo{\epsilon^{-1}}$ suffices to observe such a decrease with constant probability.
\paragraph{Pauli-diagonal setting.}
To isolate the effect of dissipation, we first analyze the case where the dissipative part is diagonal in the Pauli basis.
In this setting, the dynamics are generated by 
\[\calD_{\diag}(\rho) = \sum_{P \in \calP_n} \alpha_{P} (P\rho P - \rho),\]
with coefficients $\alpha_P \geq 0.$

The Pauli operators form an eigenbasis of this generator, so each Pauli operator evolves independently under the dynamics and can therefore be viewed as a separate ``mode" of the system, each decaying at some rate determined by the coefficients $\{\alpha_P\}$.
The key observation is that the Bell identity probability $I(t)$ averages the contribution of all these Pauli modes.
If the generator $\calD_{\diag}$ has large norm, then many Pauli modes must decay at noticeable rates, causing $I(t)$ to decrease over time.
Our analysis in Section \ref{subsec:pauli_diag} formalizes this intuition.\\\\
\textbf{Reduction to the Pauli-diagonal case via twirling.}
In general, the action of the Lindbladian mixes different Pauli operators, making it difficult to directly interpret the decay of individual operator modes. 
Moreover, direct Bell sampling of $e^{t\calL}$ is affected by coherent Hamiltonian rotations.
To isolate the dissipative contribution, we apply a standard symmetry-reduction technique known as Pauli twirling, which averages the dynamics over random Pauli conjugations.
Averaging over random Pauli or Clifford conjugations is a common symmetrization technique for quantum channels and can be justified using the theory of unitary $2$-designs \cite{dankert2009exact}.
Such twirling procedures are also widely used in quantum error correction and noise characterization, where they reduce general noise processes to Pauli or depolarizing channels \cite{knill2008randomized, wallman2016noise}.
Let the Pauli twirl of a superoperator $\Phi$ be
\[\calT(\Phi):= \EE_{P \in \calP_n} [\calU_P^{\dagger} \circ \Phi \circ \calU_P],\]
where $\calU_P = P \rho P^{\dagger}$,
then for a Lindblad generator $\calL = -i[H,\cdot ] + \calD$, the twirling operator removes the Hamiltonian component, 
\[\calT(-i[H, \cdot]) = 0,\] and therefore the reduced generator becomes 
\[\widetilde{\calL}:= \calT(\calL) = \calT(\calD).\]
Thus, in the purely Hamiltonian case $\calD = 0$, the twirled generator vanishes to $0$ and the corresponding evolution $e^{t\widetilde{L}}$ becomes the identity channel, which can be directly detected through the Bell sampling statistic mentioned above. 

The same twirling operation also removes off-diagonal couplings between distinct Pauli observables while preserving the diagonal components of the dissipative dynamics.
As a result, the twirled generator is Pauli-diagonal and falls into the setting analyzed above. 
Intuitively, this reduction converts the dissipative part of a general Lindbladian into one where each Pauli observable evolves independently, while the coherent Hamiltonian contribution has been removed.

A key step in our analysis is to show that the dissipative strength of the original generator is preserved up to constant factors  under this reduction. 
Using the locality and bounded-degree assumptions on the jump operators, we prove that the Frobenius norm of the original dissipator controls the Frobenius norm of the twirled dissipator.
Consequently, if the original dynamics contain significant dissipation, then the twirled dynamics must also contain nontrivial dissipation, which can be detected using the Bell-sampling test developed for the Pauli-diagonal setting. 
A detailed analysis is included in Section \ref{subsec:twirling}.
\paragraph{Implementing the twirled dynamics.}
One remaining issue is that our analysis relies on the dynamics generated by the twirled Lindbladian, while experimentally we only have access to the original evolution channels $e^{t\calL}$.
In general, twirling the channel is not the same as evolving under the twirled generator, so we cannot directly implement the Pauli-diagonal dynamics mentioned above.

To address this, we show that the twirled dynamics can nevertheless be approximated using short-time evolutions of the original system.
Intuitively, when the evolution time $t$ is very small, the channel $e^{t\calL}$ is well approximated by its first-order expansion $I + t\calL.$
Since twirling acts linearly on superoperators, twirling this short-time evolution approximately produces $I + t\widetilde{\calL}$, which corresponds to the first-order approximation of evolution under the twirled generator $e^{t\widetilde{\calL}}.$
By repeatedly applying such short twirled steps, we can simulate the evolution generated by $\widetilde{\calL}$ over longer times. 
Our analysis in Section \ref{subsec:trotter} shows that the accumulated error remains controlled, and that the required number of short-time steps scales only polynomially in the relevant parameters. 
Consequently, this detection procedure, fully presented in Section \ref{subsec:certification}, is able to effectively probe the Pauli-diagonal dynamics while only using black-box access to the original evolution channels.
\paragraph{Connection to purity testing.}
% The detection framework described above is not specific to the Bell-sampling statistic. 
% More generally, it applies to any measurement that detects contraction of operators under the Lindbladian evolution.
The detection framework above uses Bell sampling after a randomized reduction that removes coherent Hamiltonian evolution.
One natural alternative is to use a statistic that is intrinsically invariant under Hamiltonian dynamics, such as the purity of the Choi state associated with the channel $\calE_t = e^{t\calL}.$

Closely related notions have also been studied in the literature on unitary estimation and testing \cite{chen2023unitarity}, which aim to quantify how close a quantum channel is to unitary evolution. 
These approaches similarly capture contraction phenomena in operator space, and are typically expressed in terms of the singular values or Pauli transfer matrix of the channel.

Recall that the Choi state of a quantum channel \cite{choi1975completely, jamiolkowski1972linear} is defined as 
\[\rho_{\calE_t} := (\calE_t \otimes I)(\ket{\Phi}\bra{\Phi}),\]
where $\ket{\Phi}$ is the maximally entangled state on two $d$-dimensional registers. 
The purity of this state, $P(\calE_t) := \Tr(\rho_{\calE_t}^2)$,
is equal to the squared Hilbert–Schmidt norm of the channel and can therefore be interpreted as a measure of how strongly the channel preserves the Hilbert–Schmidt norm of operators.
Thus, it provides a natural measure of how close the channel is to a unitary evolution. 
Under purely Hamiltonian dynamics, $\calE_t$ is a unitary conjugation channel and therefore acts as an orthogonal transformation on the space of operators equipped with the Hilbert–Schmidt inner product.
In this case the corresponding Choi state remains pure for all $t$.
In contrast, dissipative dynamics introduce mixing in the operator evolution, which causes the Choi state to become mixed and its purity to decrease over time.

From a spectral perspective, the purity statistic aggregates the squared singular values of the superoperator $e^{t\calL}$.
If the dissipative component of the generator has norm at least $\epsilon$, then a nontrivial fraction of operator modes must constract at rate $\Omega(\epsilon)$.
Consequently, the purity $P(\calE_t)$ decreases by a detectable amount after evolution time $t = \bigo{\epsilon^{-1}}.$
This gives a conceptually natural route to dissipation detection: instead of removing Hamiltonian phases and then measuring a trace statistic, one can use a second-moment statistic that is already insensitive to coherent rotations.
However, this invariance alone does not automatically yield the quantitative guarantees needed here. 
To prove an $\bigo{\epsilon^{-1}}$ time scaling from a Frobenius-norm promise on the dissipator, one may still need to relate the dissipative strength to a sufficiently large collection of Pauli observables with non-negligible negative decay rates.
In our setting, this is precisely what the Pauli-twirling reduction provides: it converts the dissipative dynamics into a Pauli-diagonal form where the decay rates can be controlled using the locality and bounded-degree assumptions. 
Thus, although Choi-purity testing is conceptually simpler in its treatment of Hamiltonian evolution, achieving comparable scaling guarantees would likely still require an analysis that passes through, or closely parallels, the twirled Pauli-diagonal setting.

% The analysis underlying this statement closely parallels the Bell-sampling analysis of Section \ref{sec:decay_primitive}.
% In both cases, the key observation is that dissipation induces contraction in operator space under the Lindbladian evolution.
% Bell sampling detects this contraction by measuring the coefficient of the identity operator in the Pauli expansion of the channel, while the purity statistic instead aggregates the squared singular values of the superoperator. 
% Apart from this change in the observable, the remainder of the detection analysis carries over unchanged.

Furthermore, a key difference between this primitive and the Bell sampling statistic lies in the required experimental access.
Estimating $P(\calE_t)$ or related quantities used in unitarity estimation typically requires coherent access to multiple copies of the channel or the ability to prepare and jointly measure multiple copies of the Choi state. 
In contrast, the Bell-sampling primitive provides an analogous contraction statistic using only one-copy measurement access.
Thus, for our specific setting, the resulting procedure differs only in the measurement primitive. 
Estimating $P(\calE_t)$ requires preparing two copies of the Choi state and performing a SWAP test, which effectively measures the overlap between two independently generated channel outputs \cite{barenco1997stabilization, buhrman2001quantum}.
While conceptually simple, this requires coherent access to two copies of the channel $\calE_t$, which may not always be experimentally feasible.
For this reason, we focus on the Bell-sampling primitive in this work, which provides a comparable testing statistic while requiring more minimal experimental access.

% The resulting procedure differs only in the measurement primitive (using two copies of the channel and a swap test to estimate the purity) while the remainder of the analysis carries over unchanged.
% \YT{might be useful to have a more quantitative discussion here.}
% \YT{should mention it's not always experimentally feasible to have two copies of the same channel}

\subsection{Related work}
We briefly discuss a couple of lines of work most closely related to the detection problem studied in this paper. 
Our discussion is not meant to be comprehensive; we highlight only a few representative results and refer the reader to the cited works and references therein for a broader overview of the literature.

\paragraph{Testing Hamiltonian dynamics}
A growing line of work studies the problem of testing properties of Hamiltonian dynamics given black-box access to the induced unitary evolution.
In this setting, the goal is to distinguish whether an unknown Hamiltonian satisifes a given property or is far from it, without reconstructing the full generator.
Properties that have been studied include, for example, testing whether a Hamiltonian is $k$-local (locality testing) \cite{bluhm2026hamiltonian, kallaugher2025hamiltonian}, or whether the Hamiltonian can be expressed as a linear combination of finite Pauli operators (sparsity testing) \cite{arunachalam2025testing, sinha2025improved}.
These works establish testing as a distinct paradigm from learning, focusing on property verification using limited dynamical access, and naturally motivate closely related notion of certification.

\paragraph{Certification of Hamiltonian dynamics}
Closely related to testing are works on certification of Hamiltonian dynamics, where the goal is to verify that an unknown Hamiltonian is equal (or close) to a specified target, or to distinguish it from Hamiltonians that are far from this target.
Recent works have shown that such certification tasks can be carried out efficiently under locality assumptions \cite{bluhm2025certifying, lee2025optimal}.
In these settings, structural constraints such as locality and bounded degree play a crucial role, which enable the use of tools such as hypercontractivity and Bonami-type inequalities to relate different operator norms and certify properties of the dynamics.
These results provide prior examples in which locality assumptions allow for efficient certification from dynamical access, and are therefore consistent with the structural assumptions adopted in our work.
Our results can be viewed as extending this line of work to open-system dynamics, where the goal is to detect the absence of dissipation rather than certify a specific Hamiltonian.

% \paragraph{Certification of Hamiltonian dynamics.}
% Several works study the problem of testing properties of Hamiltonians given access to the closed-system Hamiltonian dynamics.
% In particular, optimal time scalings have been achieved for certifying whether an unknown Hamiltonian is close to a target Hamiltonian under various operator norms with varying constraints on the structure of Hamiltonians \cite{gao2026quantum, huang2023learning, dutkiewicz2024advantage, bakshi2024structure, castaneda2025hamiltonian, odake2024higher, zhao2025learning, hu2025ansatz}.
% We refer reader to Table 1 in \cite{gao2026quantum} for a comprehensive list of results for a comprehensive summary of existing results on certification of Hamiltonian dynamics in these settings.
% Our work is particularly inspired by the certification framework of \cite{lee2025optimal},  which uses Bell sampling to estimate the overlap of the unitary evolution $e^{-itH}$ with the identity channel, as well as the techniques used in the sparsity testing procedure \cite{sinha2025improved}.
% We show that this Bell-sampling primitive naturally extends to the setting of open-system dynamics, allowing us to detect dissipative components in Lindbladian evolutions.

\paragraph{Learning open-system generators.}
Recently, there has been significant progress in extending learning algorithms in the Hamiltonian setting to the open-system quantum dynamics \cite{ivashkov2026ansatz,francca2025learning, stilck2024efficient}.
For example, algorithms have been developed for learning sparse or local Lindbladian generators from dynamical measurements under structural assumptions such as locality or sparsity.
Since learning the full generator in particular allows one to detect whether dissipation is present, these results also imply certification procedures.
However, these approaches typically rely on stronger experimental access and incur higher complexity than the certification task studied in this work.
More broadly, the Hamiltonian-removal and Pauli-diagonalization reductions developed here may also serve as useful primitives for future learning algorithms, especially for tasks that seek partial information about the dissipative component rather than a full reconstruction of the Lindblad generator.
We hope that this perspective can help inform future algorithms for learning dissipative dynamics, by suggesting ways to first isolate and diagonalize the dissipative component before estimating its finer structure.

\subsection{Outlook}
Our results leave several exciting open questions and opportunities for further investigation. 
We briefly highlight a few possible directions below.

\paragraph{Relaxations of structural assumptions.}
Our analysis relies on several structural assumptions on the
Lindbladian generator, including locality and bounded overlap of the jump operators as well as a bound on the generator norm.
These assumptions ensure that the dissipative component remains detectable after the Pauli twirling reduction and allow us to control the error when approximating the twirled dynamics using short-time evolutions of the original system.
An important direction for future work is to relax these assumptions and understand whether similar certification guarantees can be obtained for more general open-system dynamics. 
In the Hamiltonian setting, optimal time scalings for the certification task under more minimal assumptions are known \cite{gao2026quantum}.
Therefore, it would be interesting to extend our results to Lindbladians with higher-degree interactions or without explicit bounds on the generator norm.

\paragraph{Dependence on locality.} 
Our complexity bounds, while optimal in scaling with the error $\epsilon$, also exhibit nontrivial dependence on structural parameters of the Lindbladian, such as the locality $k$ and the degree $\Delta$ of the jump operators.
While these parameters are assumed to be constant in our setting, the resulting query complexity and total evolution time can scale poorly with these quantities.
It would therefore be interesting to develop certification procedures with improved dependence on such structural parameters, or to identify alternative techniques that avoid these dependencies altogether.
Such improvements would broaden the applicability of our framework to more general classes of open-system dynamics.

\paragraph{Alternative norms.}
In this work we formulate the certification task using the normalized Frobenius norm (equivalently the Pauli $2$-norm) of the dissipative component.
It would be interesting to study analogous certification guarantees under other norms that capture different operational notions of distance between generators.
One challenge is that our analysis relies heavily on the relation between Bell-sampling statistics and the Frobenius norm of the generator, and it is not clear whether comparable observables exist for stronger norms such as the diamond norm.
Understanding whether similar optimal time-scaling guarantees can be obtained under these alternative metrics remains an interesting direction for future work.

\paragraph{Tolerant testing.}
Our results address a standard property testing formulation in which the goal is to distinguish purely Hamiltonian dynamics ($\calD=0$) from dynamics containing a dissipative component of magnitude at least $\epsilon$.
A natural direction is to develop \emph{tolerant} certification procedures that distinguish between the cases $\|\calD\| \le \epsilon_1$ and $\|\calD\| \ge \epsilon_2$ for $0 \le \epsilon_1 < \epsilon_2 $.

\section{Preliminaries}
We consider an $n$-qubit system with Hilbert space $\calH := \C^{\otimes n}$ and dimension $d := 2^n.$
We write $\calP_n := \{I,X,Y,Z\}^{\otimes n}$ for the set of $n$-qubit Pauli strings.
For $P \in \calP_n$, its Pauli weight $|P|$ is the number qubits acted on by non-identity Pauli operators. 
For a subset $S \subseteq [n]$, we say that an operator is supported on $S$ if it acts nontrivially only on qubits in $S$.
Lastly, we write $x \sim \mathcal{X}$ to denote uniform sampling from a finite set $\mathcal{X}$.

\subsection{Superoperators}\label{subsect:superop}
Let $\mathsf{L}(\C^d)$ denote the space of linear operators on $\C^d.$
We equip $\mathsf{L}(\C^d)$ with the Hilbert-Schmidt inner product 
\[\langle X, Y \rangle_{\text{HS}} := \Tr(X^\dagger Y), \quad X,Y \in \mathsf{L}(\C^d),\]
where $\Tr(\cdot)$ denotes the usual operator trace on $d\times d$ matrices.
A superoperator (i.e. a quantum channel) is a complete-positive, trace-preserving (CPTP) linear map of operators, denoted by $\calE: \mathsf{L}(\C^d) \to \mathsf{L}(\C^d)$ and can be viewed as a $d^2 \times d^2$ linear operator acting on the Hilbert-Schmidt space.
We define the trace of a superoperator $\calE$ by
\begin{equation}
    \Tr(\calE) := \sum_{a=1}^{d^2}\langle F_a, \calE(F_a) \rangle_{\text{HS}} = \sum_{a=1}^{d^2} \Tr(F_a^\dagger \calE(F_a)),
\end{equation}
where $\{F_a\}_{a=1}^{d^2} \subset \mathsf{L}(\C^d)$ is the set of any orthonormal operator basis with respect to $\langle \cdot, \cdot \rangle_{\text{HS}}.$
This definition is independent of the choice of orthonormal operator basis and coincides with the ordinary matrix trace of $\calE$ when represented as a linear operator on the Hilbert–Schmidt space.
For concreteness, choosing the matrix-unit basis $F_{ij} := \ket{i}\bra{j}$ with $i,j \in[d]$, which satisfies $\langle F_{ij}, F_{kl} \rangle = \delta_{ik} \delta_{jl}$, gives us the explicit expression
\begin{equation}
    \Tr(\calE) = \sum_{i,j=1}^d\Tr(\ket{j}\bra{i}\calE (\ket{i}\bra{j})) = \sum_{i,j=1}^d \bra{i} \calE(\ket{i}\bra{j}) \ket{j}.
\end{equation}
We will repeatedly make use of the normalized trace 
\begin{equation}
    I(\calE) := \frac{1}{d^2}\Tr(\calE),
\end{equation}
which extracts the coefficient of the identity superoperator in the Hilbert–Schmidt expansion of $\calE$.
This quantity will later be evaluated on time-evolution channels arising from Lindbladian dynamics.

We define the normalized Frobenius norm of a superoperator $\calE$ by 
\begin{equation}\label{eqn:norm_frob}
    \|\calE\|_F := \sqrt{\frac{1}{d^2}\Tr(\calE^\dagger \calE}).
\end{equation}
Equivalently, for any Hilbert–Schmidt orthonormal operator basis $\{F_a\}_{a=1}^{d^2}$,
\[\|\calE\|_F = \frac{1}{d^2}\sum_{a=1}^{d^2}\Tr(\calE(F_a)^\dagger \calE(F_a)).\]
This normalization ensures that the identity superoperator $\calI$ satisfies $\|\calI\|_F = 1$, and it allows for dimension-independent comparisons of superoperators.
Additionally, this norm is invariant under unitary conjugations.

We will also make use of the diamond norm of a superoperator, which measures the maximum distinguishability it can induce when acting on part of an entangled system.
For a superoperator $\calE: \mathsf{L}(\C^d) \to \mathsf{L}(\C^d)$, the diamond is defined as 
\begin{equation}
\|\mathcal{E}\|_{\diamond}:= 
\underset{
\substack{
\rho \in \mathsf{L}(\mathbb{C}^d \otimes \mathbb{C}^d), \\
\|\rho\|_1 = 1}
}{\sup}
\|(I \otimes \mathcal{E})(\rho)\|_1,
\end{equation}
where $\|\cdot\|_1$ denotes the trace norm and $I$ is the identity map on $\mathsf{L}(\C^d)$.
The diamond norm satisfies the submultiplicativity property 
\begin{equation}
    \|\calE_1 \circ \calE_2 \|_{\diamond} \leq \|\calE_1\|_{\diamond} \|\calE_2\|_{\diamond}
\end{equation}
and for every quantum channel $\Phi$, which is a completely positive trace-persevering map, $\|\Phi\|_{\diamond} = 1.$

\subsection{Lindbladian dynamics}\label{subsec:lindb}
We consider Markovian open-system dynamics on $n$ qubits generated by a time-independent Lindbladian superoperator
\[\calL: \mathsf{L}(\calH) \to \mathsf{L}(\calH).\]
Throughout, we write 
\begin{equation}
    \calL(\rho) = -i[H,\rho]  + \underset{\calD(\rho)}{\underbrace{\sum_{a=1}^m L_a\rho L_a^\dagger - \frac{1}{2}\{L_a^\dagger L_a, \rho\}}},
\end{equation}
where $H = H^\dagger$ is the Hamiltonian and $\calD$ is the dissipative part of the dynamics.
Lindbladians describe open quantum systems in which environmental noise acts continuously and without memory, leading to irreversible dynamics.
Throughout this work, we restrict attention to Lindbladians whose dissipative part admits a local structure.
The dissipator admits a Lindblad form
\begin{equation}
    \calD(\rho) = \sum_{a=1}^m \Big( L_a \rho L_a^\dagger - \tfrac12\{L_a^\dagger L_a,\rho\} \Big)
\end{equation}
for a collection of jump operators $\{L_a\}_{a=1}^m$.

The generator $\calL$ defines a one-parameter family of quantum channels 
\begin{equation}
    \calE_t := e^{t\mathcal L}, \qquad t \ge 0,
\end{equation}
which are completely positive and trace preserving, and satisfy the semigroup property
$\calE_{t+s} = \calE_t \circ \calE_s$.
When $\calD = 0$, the evolution is purely coming from the Hamiltonian $H$ and $\calE_t$ is a unitary conjugation channel for all $t$.
In contrast, nonzero $\calD$ gives rise to irreversible effects such as decoherence.
Thus, we restrict to $t \geq 0$ as the evolution is not invertible in general.

Without loss of generality, we assume that all jump operators are traceless.
Indeed, any identity component of a jump operator can be absorbed into the Hamiltonian part of the generator and does not contribute to dissipation, and this can be achieved without changing the physical dynamics.
This choice fixes a convenient representation of the Lindbladian for us to work with. 

We assume locality at the level of jump operators together with bounded overlap.
More specifically, we say that a jump operator $L_a$ is \emph{$k$-local} if it acts nontrivially on at most $k$ qubits.
In addition, we impose a bounded-degree condition: the collection of jump operators $\{L_a\}_{a=1}^m$ has \emph{degree at most $\Delta$} if each qubit participates in the support of at most $\Delta$ distinct jump operators.
The precise description is stated in Definition \ref{def:locality_and_degree} below.
These assumptions are standard in the study of local open-system dynamics and capture the physically natural setting in which noise acts locally and with bounded overlap.

\begin{definition}[$k$-local jump operators of degree $\Delta$]\label{def:locality_and_degree}
    For a Pauli string $P \in \calP_n \in \{I,X,Y,Z\}^{\otimes n}$, let  
    $L_a = \sum_{P \in \calP_n}\gamma_{a,P}P$
    denote the Pauli expansion of the jump operator $L_a$.
    We say the collection of jump operators $\{L_a\}_a$ is $k$-local if for every $a$, there exists a subset $S_a \subseteq [n]$ with $S_a \leq k$ such that $\gamma_{a,P} \neq 0$ if and only if $\supp(P) \subseteq S_a$.
    Furthermore, they are of degree at most $\Delta$ if every qubit belongs to the support of at most $\Delta$ jump operators: $\max_{i \in [n]}|\{a: i \in S_a\}|\leq \Delta.$
\end{definition}

We note that the particular representation of $\calD$ as written above is not unique.
Our results depend only on the superoperator $\calD$ itself and on the existence of some representation satisfying the above locality and degree conditions.

In addition to the locality assumptions above, we assume that the generator $\calL$ has bounded strength. 
Specifically, we assume that the Lindbladian superoperator
satisfies
\begin{equation}
    \|\mathcal L\|_{\diamond} \le L ,
\end{equation}
where $\|\cdot\|_{\diamond}$ denotes the diamond norm on superoperators
(defined in Section \ref{subsect:superop}).  
Physically, this corresponds to assuming that the total rate of the dissipative and Hamiltonian processes generating the evolution is bounded. 
This assumption ensures that the short-time expansion of the channel $\mathcal E_t = e^{t\mathcal L}$ is well controlled, which will be required in the perturbative analysis of the detection procedure.

\subsection{Bell sampling for quantum channels}\label{subsec:bell_sampling}
We now describe a simple experimental procedure, known as Bell sampling \cite{montanaro2017learning, hangleiter2024bell}, that allows us to extract a single scalar quantity from an unknown quantum channel. 
This procedure will serve as the basic measurement primitive used throughout the paper.

Let $\sA$ and $\sB$ be two $n$-qubit registers, and let 
\[\ket{\Phi}_{\sA\sB}:= \frac{1}{\sqrt{d}} \sum_{z \in \{0,1\}^n} \ket{z}_\sA\ket{z}_{\sB}\] be the maximally entangled state.
For a quantum channel $\calE: \mathsf{L}(\C^d) \to \mathsf{L}(\C^d)$, we define its Choi state by 
\[\rho_{\calE}^{\sA\sB} := (\calE \otimes I) (\ket{\Phi} \bra{\Phi}).\]
The Bell basis on $\mathsf{AB}$ is given by 
\[\{\ket{\Phi_s}: s\in\{0,1\}^{2n}\},\]
where each $\ket{\Phi_s}$ is obtained from $\ket{\Phi}$ by applying a Pauli operator on one register.
Bell sampling applied to the channel $\calE$ consists of preparing the Choi state $\rho_{\calE}^{\sA\sB}$ and measureing $\sA\sB$ in the $n$-fold Bell basis $\{\ket{\Phi_s}\}_{s \in \{0,1\}^{2n}}.$
The resulting outcome $s$ is distributed according to 
\begin{equation}
    q_{\calE}(s) := \Pr[\text{Bell outcome } s \text{ on } \rho_{\calE}^{\sA\sB}] = \bra{\Phi_s} \rho_{\calE}^{\sA\sB} \ket{\Phi_s}.
\end{equation}
Of particular importance is the probability of obtaining the identity outcome, which we denote by 
\begin{equation}
    I(\calE) := \Pr[\text{Bell outcome } \ket{\Phi}] = \bra{\Phi} \rho_{\calE}^{\sA\sB} \ket{\Phi} = \frac{1}{d^2}\Tr(\calE),
\end{equation}
where the last equivalence can be seen through Lemma \ref{lem:bell_identity} below.

\begin{lemma}[Bell identity for the superoperator trace]\label{lem:bell_identity}
    Let $\calE:\mathsf{L}(\mathbb{C}^d)\to \mathsf{L}(\mathbb{C}^d)$ be a linear map. Define
\[
|\Phi\rangle := \frac{1}{\sqrt d}\sum_{i=1}^d |i\rangle|i\rangle,
\qquad
I(\mathcal{E}) := \langle \Phi|(\mathcal{E}\otimes I)(|\Phi\rangle\langle \Phi|)|\Phi\rangle.
\]
Then
\[
I(\mathcal{E})=\frac{1}{d^2}\Tr(\mathcal{E}),
\]
where $\Tr(\mathcal{E})$ denotes the trace of $\mathcal{E}$ as a linear operator on the $d^2$-dimensional
Hilbert--Schmidt space.
\end{lemma}

\begin{proof}
    First, we expand
    \begin{equation}
        |\Phi\rangle\langle\Phi|
    =\frac{1}{d}\sum_{i,j=1}^d |i\rangle\langle j|\otimes |i\rangle\langle j|.
    \end{equation}
    By linearity, we have that 
    \begin{align}
        (\mathcal{E}\otimes I)(|\Phi\rangle\langle\Phi|)
    =\frac{1}{d}\sum_{i,j=1}^d \mathcal{E}(|i\rangle\langle j|)\otimes |i\rangle\langle j|.
    \end{align}
    For any $X,Y\in\mathsf{L}(\C^d)$, the following holds:
    \begin{align}
        \bra{\Phi} (X\otimes Y) \ket{\Phi} &= \Tr\left[(X\otimes Y) \ket{\Phi}\bra{\Phi}\right]\\
        &= \frac{1}{d} \sum_{i,j} \Tr(X\ket{i}\bra{j})\Tr(Y\ket{i}\bra{j})\\
        &= \frac{1}{d} \sum_{i,j} \bra{j}X\ket{i} \bra{j}Y\ket{i}\\
        &= \frac{1}{d}\Tr(XY^T).
    \end{align}
    Choosing $X = \calE(\ket{i}\bra{j})$ and $Y = \ket{i}\bra{j}$, we obtain 
    \[\bra{\Phi}(\calE \otimes I) (\ket{\Phi} \bra{\Phi})\ket{\Phi} = \frac{1}{d}\Tr(\calE(\ket{i}\bra{j}) \ket{j}\bra{i}) = \frac{1}{d} \bra{i} \calE(\ket{i}\bra{j}) \ket{j}.\]
    Putting everything together, we arrive at the desired equivalence
\begin{align}
    I(\calE) = \frac{1}{d}\sum_{i,j} \frac{1}{d}\bra{i} \calE(\ket{i}\bra{j})\ket{j} = \frac{1}{d^2} \sum_{i,j} \bra{i} \calE(\ket{i}\bra{j}) \ket{j} = \frac{1}{d^2} \Tr(\calE).
\end{align}
\end{proof}

In this work, we apply Bell sampling to channels arising from continuous-time open-system dynamics. 
Specifically, given a time-independent Lindblad generator $\calL$ as defined in Section \ref{subsec:lindb}, we consider the associated time-$t$ evolution channel $\calE_t := e^{t\calL}.$
We will therefore view the Bell identity probability as a function of time, and write 
\begin{equation}
    I(t) := I(\calE_t) = \frac{1}{d^2}\Tr(e^{t\calL}).
\end{equation}
Intuitively, $I(t)$ measures the average overlap between an operator and its image under the evolution, or equivalently how strongly the channel resembles the identity superoperator.
When the relevant dynamics are diagonal with non-positive real eigenvalues, this overlap becomes an average of decay factors and therefore captures the average survival of the corresponding operator components.
If $\eta_1, \ldots, \eta_{d^2}$ are the eigenvalues of $\calL$, counted with algebraic multiplicity, then $e^{t\calL}$ has eigenvalues $\{e^{t\eta_a}\}_{a \in [d^2]}$ and $\Tr(e^{t\calL}) = \sum_a e^{t\eta_a}$.
This remains valid even when $\calL$ is not diagonalizable, by considering its Jordan decomposition. 
We include a detailed analysis of this in Lemma \ref{lem:jordan_form} below.
We note that $I(t)$ is real, 
since $\calL$ preserves Hermiticity and thus its eigenvalues occur in complex-conjugate pairs (indeed, if $\calL(X) = \lambda X$, then $\calL(X^\dagger) = \calL(X)^\dagger =\bar{\lambda}X^\dagger$).

\begin{lemma}\label{lem:jordan_form}
    Let $\calL$ be a linear operator acting on a finite-dimensional complex vector space.
    Let $\eta_1, \ldots, \eta_{d^2}$ denote the eigenvalues of $\calL$, counted with algebraic multiplicity.
    Then for all $t \geq 0$, 
    \begin{equation}
        \Tr(e^{t\calL}) = \sum_{a=1}^{d^2} e^{t\eta_a},
    \end{equation}
    and consequently,
    \begin{equation}
        I(t) = \frac{1}{d^2}\sum_{a=1}^{d^2}e^{t\eta_a}.
    \end{equation}
\end{lemma}

\begin{proof}
    Over $\C$, there exists an invertible $S$ such that 
\[\calL = SJS^{-1},\]
where $J$ is a block diagonal matrix wit Jordan blocks $J = \bigoplus_{a} J_a, J_a = \eta_a I_{m_a} + N_a.$
We use $\eta_a$ to denote the eigenvalues of $\calL$, $m_a$ to denote the size of the Jordan block associated with index $a$, and $N_a$ a nilpontent matrix with ones on the superdiagonal and zeros everywhere else, with $N_a^{m_a} = 0.$
Exponentiating $\calL$, we obtain 
\begin{align}
    e^{t\calL} = e^{tSJS^{-1}} 
    = Se^{tJ} S^{-1} 
    = S\left(\bigoplus_{a} e^{tJ_a}\right) S^{-1}.
\end{align}
Because $\eta_a I_{m_a}$ commutes with $N_a$, 
\[e^{tJ_a} = e^{t(\eta_a I_{m_a} + N_a)} = e^{t\eta_a I_{m_a}} e^{tN_a} = e^{t\eta_a} e^{tN_a}.\]
To compute $\Tr(e^{tJ_a})$, we have
\[\Tr(e^{tJ_a}) = \Tr(e^{t\eta_a} e^{t N_a}) = e^{t\eta_a} \Tr(e^{tN_a}).\]
We can expand $e^{tN_a}$ as a finite series $e^{tN_a} = \sum_{k=0}^{m_a - 1} \frac{t^k}{k!} N_a^k$ and compute 
\[\Tr(e^{tN_a}) = \sum_{k=0}^{m_a - 1} \frac{t^k}{k!} \Tr(N_a^k).\]
Note that for $k \geq 1$, $N_a^k$ is strictly upper triangular, thereofre it has zero diagonal and $\Tr(N_a^k) = 0.$
Also $\Tr(N_a^0) = \Tr(I_{m_a}) = m_a,$
so $\Tr(e^{tN_a}) = m_a$ and $\Tr(e^{tJ_a}) = m_a e^{t\eta_a}.$
Because $e^{tJ}$ is block diagonal,
\[\Tr(e^{tJ}) = \sum_a \Tr(e^{tJ_a}) = \sum_a m_a e^{t\eta_a}.\]
Applying trace invariance, 
\[\Tr(e^{t\calL}) = \Tr(Se^{tJ}S^{-1}) = \Tr(e^{tJ}) = \sum_a m_a e^{t\eta_a} = \sum_{b=1}^{d^2} e^{t\eta_b},\]
where the last equality is simply listing the eigenvalues with algebraic multiplicity. 
\end{proof}

In later sections we will compare the Bell sampling statistic of different quantum channels that are close to each other in Frobenius norm.  
It will therefore be useful to relate the difference in the identity outcome probability $I(\mathcal E)$ to the distance between superoperators.  
The following simple lemma shows that the diamond norm controls the deviation of the Bell identity probability for arbitrary superoperators.

\begin{lemma}[Stability of the Bell identity probability] \label{lem:diff_bell_identity_prob}
For any superoperators $\mathcal E$ and $\mathcal F$,
\begin{equation}
    \big| I(\mathcal E) - I(\mathcal F) \big|
\le \frac{1}{2}
\|\mathcal E - \mathcal F\|_{\diamond} .
\end{equation}
\end{lemma}

\begin{proof}
Let $\ket{\Phi}$ denote the maximally entangled state on two $d$-dimensional registers. 
By definition of the Bell identity probability,
$I(\mathcal{E}) := \Tr(\ket{\Phi}\langle \Phi|(\mathcal{E}\otimes I)(|\Phi\rangle\langle \Phi|)).$
Therefore, we obtain 
\begin{align}
    |I(\mathcal E)-I(\mathcal F)| &= \left|\Tr(\ket{\Phi}\bra{\Phi} (I \otimes (\calE -\calF)) (\ket{\Phi}\bra{\Phi}))\right| \\
    &\leq \frac{1}{2} \| (I \otimes (\calE - \calF)) (\ket{\Phi}\bra{\Phi}) \|_1\\
    &\leq \frac{1}{2} \|\calE - \calF\|_{\diamond},
\end{align}
where the first inequality uses $0 \leq \ket{\Phi}\bra{\Phi} \leq I$ and the variational characterization of the trace norm, and the second follows from the definition of the diamond norm.
% Directly applying the definition of Bell identity probability, we obtain
% \[
% |I(\mathcal E)-I(\mathcal F)|
% =
% \frac{1}{d^2}|\mathrm{Tr}(\mathcal E-\mathcal F)|.
% \]
% Using the Cauchy–Schwarz inequality for the Frobenius inner product,
% \[
% |\mathrm{Tr}(A)| \le d^2 \|A\|_F ,
% \]
% which yields the claimed bound.
\end{proof}

\subsection{Probability tools}
We collect a few known inequalities in probability theory that will be of use in later analysis. 

\begin{lemma}[Paley-Zygmund]\label{lem:paley-zygmund}
    Let $Z \geq 0$ be a random variable with $\EE[Z^2] < \infty.$
    Then for every $\theta \in (0,1), $
    \begin{equation}
        \Pr[Z \geq \theta \EE[Z]] \geq (1-\theta)^2 \frac{\EE[Z]^2}{\EE[Z^2]}.
    \end{equation}    
\end{lemma}

\begin{lemma}[Bonami–Beckner hypercontractivity] \label{lem:hypercontract}
Let $f: \{-1, 1\}^n \to \R$ have degree at most $k$.
Then for every $q \geq 2$,
\begin{equation}
    \|f\|_q  \leq (q-1)^{k/2}\|f\|_2,
\end{equation}
where $\|f\|_q:= (\EE_{x\sim \{-1,1\}^n}[|f(x)|^q])^{1/q}.$ 
\end{lemma}
While our analysis operates on bitstrings $b \in \{0,1\}^n$, for analytical convenience we sometimes identify $b$ with $x \in \{-1,1\}^n$ via the mapping $x_i = (-1)^{b_i}.$
This identification preserves the uniform product measure and all moments, and it is used implicitly in applications of hypercontractive inequalities throughout this work.

\section{A Bell-sampling primitive for spectral decay
}\label{sec:decay_primitive}
In this section, we introduce a simple spectral quantity associated with a Lindblad generator that serves as a building block for our later detection procedure.
The primitive captures a one-sided implication: if a non-negligible fraction of the eigenvalues of a generator have real part at most $- \epsilon$, then, for a randomly chosen evolution time $t = \bigo{\epsilon^{-1}}$, the normalized trace of the corresponding channel $e^{t\calL}$ is detectably smaller than $1$ with probability.
Since this normalized trace is exactly the Bell identity probability, Bell sampling gives an operational way to witness this spectral decay.

% Informally, this quantity measures the fraction of eigenvalues of the generator that decay at a non-negligible rate under the Lindbladian evolution.
% We show that if a constant fraction of eigenvalues exhibits decay above a threshold $\epsilon$, then a randomly chosen evolution time $t = \bigo{\epsilon^{-1}}$ suffices to witness a detectable contraction in the corresponding quantum channel with constant probability.
% This observation will form the basic testing primitive underlying our detection procedures in later sections.

Concretely, let $\calL$ be a Lindblad generator on a $d$-dimensional system, viewed as a linear operator on the $d^2$-dimensional Hilbert-Schmidt space of operators. 
Let $\eta_1, \ldots, \eta_{d^2} \in \C$ denote the eigenvalues of $\calL$, counted with algebraic multiplicity.
For $\epsilon > 0$, define 
\begin{equation}
    \Lambda (\calL, \epsilon) := \frac{1}{d^2} \sum_{a = 1}^{d^2} \mathds{1}\{-\Re(\eta_a) \geq \epsilon\}, 
\end{equation}
which denotes the fraction of eigenmodes whose amplitudes decay at rate at least $\epsilon$.
We also define, for $t \geq 0$, 
\begin{equation}
    I(t) := \frac{1}{d^2}\Tr(e^{t\calL}) = \frac{1}{d^2} \sum_{a=1}^{d^2} e^{t\eta_a}.
\end{equation}
This quantity is the uniform average of the eigenvalues $e^{t\eta_a}$ of the channel $e^{t\calL}$ viewed on the Hilbert–Schmidt space.
% For purely Hamiltonain dynamics, all eigenvalues are imaginary and thus $I(t) = 1$ for all $t$.
% When dissipation is present, the decaying modes with $\Re(\eta_a) < 0$ contribute factors of magnitude $e^{t\Re(\eta_a)} < 1$, so $I(t)$ decreases when a nontrivial fraction of modes decay.
Crucially, the quantity $I(t)$ admits a direct operational interpretation via Bell sampling, as described in Section \ref{subsec:bell_sampling}.
Given black-box access to the channel $\calE_t:=e^{t\calL}$, Bell sampling estimates the probability that $\calE_t$ acts as the identity on a maximally entangled state, which is exactly $I(t)$, and therefore can be efficiently estimated from experimental data.    
% As a result, any detectable contraction in $I(t)$ directly yields a test for the presence of dissipation.

The following lemma shows that if a constant fraction of eigenvalues have real part at most $-\epsilon$, then choosing a random evolution time $t = \bigo{\epsilon^{-1}}$ suffices to witness a detectable contraction in $I(t)$ with constant probability.

\begin{lemma}\label{lem:success_prob}
    Let $\calL$ be a time-independent Lindblad generator on $n$ qubits, and let $d = 2^n.$
    For $\epsilon > 0$, suppose that  $\Lambda:= \Lambda (\calL, \epsilon) > 0$. 
    Then for $t$ uniformly sampled from $[0, 2/\epsilon]$,
    \begin{equation}
        \Pr_t\left[I(t) \leq 1-\frac{2\Lambda}{3}\right] > \frac{2}{5}.
    \end{equation}
\end{lemma}
\begin{proof}
    Let $a$ be a uniformly random index in $[d^2]$, so $I(t) = \EE_a[e^{t\eta_a}].$
    Consider the set $F := \{a \in [d^2]: -\Re(\eta_a) \geq \epsilon\}.$
    Then $\Pr_a[a \in F] = \Lambda.$
    For every $a \in F$, we have that $|\eta_a|\geq \epsilon$, $\Re(\eta_a) \leq -\epsilon$, so $|e^{t\eta_a}| \leq e^{-\epsilon t}.$
    For $a \notin F$, it follows that $\Re(\eta_a) \leq 0$, so $|e^{t\eta_a}| = e^{t\Re(\eta_a)} \leq 1.$
    We can then obtain that for all $t \geq 0$,
    \begin{equation}
        I(t) \leq (1-\Lambda) \cdot 1 + \Lambda \EE[e^{t\eta_a}|F] \leq (1-\Lambda) + \Lambda \cdot e^{-\epsilon t} = 1 - \Lambda (1 - e^{-\epsilon t}).
    \end{equation}
    On the event that $e^{-\epsilon t} \leq \frac{1}{3}$, the above bound gives 
    \begin{equation}\label{eqn:I_bound}
        I(t) \leq (1 - \Lambda) + \Lambda \cdot \frac{1}{3} = 1 - \frac{2 \Lambda}{3}.
    \end{equation}
    Since $e^{-\epsilon t}$ is decreasing for $t \geq 0$, 
    \[e^{-\epsilon t} \leq \frac{1}{3} \iff t \geq \frac{1}{\epsilon}\log 3.\]
    Therefore, for uniformly sample $t \sim [0, 2/\epsilon]$, 
    \begin{equation}
        \Pr_t\left[e^{-\epsilon t} \leq \frac{1}{3}\right] = \Pr_t\left[ t \geq \frac{1}{\epsilon}\log 3\right] = \frac{\frac{2}{\epsilon} - \frac{1}{\epsilon}\log 3}{\frac{2}{\epsilon}} = 1 - \frac{1}{2} \log 3 > \frac{2}{5}.
    \end{equation}
    Combining with Equation \ref{eqn:I_bound}, we conclude that 
    \begin{equation}
        \Pr_t \left[I(t) \leq 1 - \frac{2\Lambda}{3}\right] > \frac{2}{5}.
    \end{equation}
    
    % we have 
    % \[\EE_t[e^{-\epsilon t}] = \frac{\epsilon}{2} \int_{0}^{2/\epsilon} e^{-\epsilon t} dt = \frac{1 - e^{-2}}{2} < \frac{1}{2}.\]
    % Since $e^{-st} \in [e^{-2},1] \subset [0,1]$ and $\EE[e^{-\epsilon t}] < \frac{1}{2}$, applying the Markov inequality gives us 
    % \[\Pr_t\left[e^{-\epsilon t} \geq \frac{3}{4} \right] \leq \frac{\EE[e^{-\epsilon t}]}{3/4} < \frac{2}{3},\]
    % and rearranging, we obtain 
    % \[\Pr_t\left[e^{-\epsilon t} \leq \frac{3}{4} \right] \geq \frac{1}{3}.\]
    % Lastly, on the event $e^{-\epsilon t} \leq \frac{3}{4}$, which happens with probability at least $\frac{1}{3}$, we have that 
    % \[I(t) \leq (1-\Lambda) + \frac{3}{4} \Lambda \leq 1 - \frac{c}{4}.\]
\end{proof}

\section{Certification of dissipation absence}
\subsection{Pauli-diagonal dissipative dynamics}\label{subsec:pauli_diag}
We first analyze the detection problem for Pauli-diagonal dissipative generators.
This setting captures the essential decay mechanism underlying our later reduction: once the coherent Hamiltonian contribution has been removed, the remaining dynamics are diagonal in the Pauli basis and can be analyzed directly through the eigenvalues associated with Pauli operators.
In particular, we show that a nontrivial Frobenius norm of the diagonal dissipative generator forces a constant fraction of Pauli operators to decay at a comparable rate.
This spectral decay can then be detected using the Bell-sampling primitive introduced in Section \ref{sec:decay_primitive}.

We consider Pauli-diagonal dissipative generators of the form
\begin{equation}
\mathcal D_{\diag}(\rho)
= \sum_{\substack{P \in \mathcal P_n \\ 1 \le |P| \le k}}
\alpha_P \,(P\rho P - \rho),
\qquad \alpha_P \in  \R_{\geq 0} .
\end{equation}
Since conjugation by a Pauli either preserves or flips another Pauli operator, the Pauli basis diagonalizes $\calD_{\diag}.$
For each $Q \in \calP_n$, the action of $\calD_{\diag}$ on $Q$ is scalar multiplication by a real non-positive eigenvalue.

\begin{lemma}\label{lem:pauli-diag}
    Let $\calD_{\diag}$ be a Pauli-diagonal dissipative generator of the form 
    \[\calD_{\diag}(\rho):= \sum_{P\in \calP_n: 1\leq |P| \leq k} \alpha_p (P\rho P - \rho), \qquad \alpha_P \in \R_{\geq 0},\]
    then for every $r \in [0,1]$, 
    \begin{equation}
        \Lambda(\calD_{\diag}, r\|\calD_{\diag}\|_F) \geq (1-r^2)^2 \cdot 9^{-k}.
    \end{equation}
\end{lemma}

\begin{proof}
    Fix $Q \in \calP_n$. 
    For any $P \in \calP_n$, conjugating by $P$ either preserves or flips $Q$:
    \[PQP = \chi(P,Q)Q, \]
    where $\chi(P,Q) = +1$ if $[P,Q] = 0$ and $\chi(P,Q) = -1$ if $\{P,Q\} = 0.$
    Therefore, 
    \begin{equation}
        \calD_{\diag}(Q) = \sum_P \alpha_P (PQP - Q) = \sum_P \alpha_P (\chi(P,Q)-1)Q.
    \end{equation}
    Then, the eigenvalue is 
    \[\eta_Q := \sum_P \alpha_P(\chi(P,Q)-1) = f(Q) - \Gamma,\]
    where we denote $f(Q):= \sum_P \alpha_P\chi(P,Q)$ and $\Gamma:= \sum_P \alpha_P.$
    Since each $\alpha_P \geq 0$ and $(1-\chi(P,Q)) \in \{0,2\}$, we have $\eta_Q \leq 0$ and thus define 
    \[F(Q):= -\eta_Q = \sum_P \alpha_P (1-\chi(P,Q)) = \Gamma - f(Q)  \geq 0.\]
    With this definition, we obtain 
    \begin{equation}
        \Lambda (\calD_{\diag}, \epsilon) = \frac{1}{4^n} \sum_{Q \in \calP_n}\mathds{1}\{F(Q) \geq \epsilon\} =  \Pr_{Q}[F(Q) \geq \epsilon].
    \end{equation}
    The normalized Frobenius norm of $\calD_{\diag}$ is 
    \begin{equation}
        \|\calD_{\diag}\|_F^2 = \frac{1}{d^2} \sum_{Q \in \calP_n} \eta_Q^2 = \EE_Q\left[\eta_Q^2\right] = \EE_Q[F(Q)^2].
    \end{equation}
    We next compute the second moment explicitly. 
    First note that for any fixed non-identity Pauli $P$, we have $\EE_Q[\chi(P,Q)] = 0.$
    Indeed, pick any site $i$ where $P_i \neq I$.
    If we sample $Q_i$ uniformly random from $\{I,X,Y,Z\}$ while keeping all other coordinates fixed, then $Q_i$ anticommutes with $P_i$ with probability $1/2$ and commutes with probability $1/2$, so the conditional expectation of $\chi(P,Q)$ over $Q_i$ is $0$ and accordingly the full expectation is $0$ as well. 
    Similarly, for $P \neq P'$, 
    \begin{equation}
        \EE_Q[\chi(P,Q)\chi(P',Q)] = 0.
    \end{equation}
    Choose a site $i$ where $P_i \neq P_i'$. 
    Then $\chi(P,Q)\chi(P',Q)$ depends on $Q_i$ through a nontrivial $\{\pm1\}$-valued factor whose average over $Q_i$ is $0$, by the same argument as above. 
    Lastly, we note that it is straightforward to see that 
    \[\EE_Q[\chi(P,Q)^2] = 1.\]
    Combining these together, we get that 
    \begin{align}
        \|\calD_{\diag}\|_F^2 &= \EE_Q[(\Gamma - f(Q))^2] \\
        &= \Gamma^2 - 2\Gamma \EE_Q[f(Q)] + \EE_Q[f(Q)^2] \\
        &= \Gamma^2 + \sum_{P,P'}\alpha_P \alpha_{P'}\EE_Q[\chi(P,Q)\chi(P',Q)] \\
        &= \Gamma^2 + \sum_P \alpha_P^2.
    \end{align}

    Since commutation is determined site-by-site, the sign $\chi(P,Q)$ factorizes across qubits.
    If $|P| \leq k$, then $\chi(P,Q)$ depends only on the $k$ qubits in the support of $P$. 
    Therefore, each term in $F(Q)$ depends on at most $k$ qubits and has Fourier degree at most $k$. 
    Thus, by the standard hypercontractivity argument in Lemma \ref{lem:hypercontract}, 
    \begin{align}
        \EE_Q[F(Q)^4] &= \|F\|_4^4 \\
        &\leq (3^{k/2} \|F\|_2)^4\\
        &= 9^k (\EE[F(Q)^2])^2 \\
        &= 9^k \|\calD_{\diag}\|_F^4.
    \end{align}

    Applying Paley-Zygmund in Lemma \ref{lem:paley-zygmund} to the nonnegative random variable $ F(Q)^2$, we get that for any $\theta \in [0,1]$,
    \begin{equation}
        \Pr[F(Q)^2 \geq \theta \EE[F(Q)^2]] \geq (1-\theta)^2 \frac{\EE[F(Q)^2]^2}{\EE[F(Q)^4]} \geq (1-\theta)^2 \cdot 9^{-k}.
    \end{equation}
    Since $F(Q) \geq 0$, this is equivalent to 
    \begin{equation}
        \Pr[F(Q) \geq \sqrt{\theta \EE[F(Q)^2]}] \geq (1-\theta)^2 \cdot 9^{-k}.
    \end{equation}
    Using the fact that $\|\calD_{\diag}\|_F = \sqrt{\EE[F(Q)^2]}$ and choosing $r:= \sqrt{\theta}$, we obtain 
    \begin{equation}
        \Pr[F(Q) \geq r \|\calD_{\diag}\|_F] \geq (1-r^2)^2\cdot 9^{-k}.
    \end{equation}
    Finally, recalling that $\Lambda(\calD_{\diag}, \epsilon) = \Pr_Q[F(Q) \geq \epsilon],$ we conclude 
    \begin{equation}
        \Lambda(\calD_{\diag}, r\|\calD_{\diag}\|_F) \geq (1-r^2)^2 \cdot 9^{-k},
    \end{equation}
    as desired.
\end{proof}
We note that combining the above Lemma \ref{lem:pauli-diag} with the Bell-sampling decay test of Lemma \ref{lem:success_prob} immediately yields an efficient detection procedure for Pauli-diagonal Lindbladians, with total evolution time $O(\|\calD_{\diag}\|_F^{-1}).$

\subsection{Random twirling}\label{subsec:twirling}
The Pauli-diagonal setting captures the core mechanism of our detection procedure, but a general Lindblad generator needs not be diagonal in the Pauli basis and may also contain a Hamiltonian component. 
A natural way to isolate the dissipative contribution and suppress the off-diagonal structure is to average over conjugation by random Pauli operators.  
This leads to the Pauli twirling operation, which projects a superoperator onto its Pauli-diagonal component and removes the Hamiltonian part.
More specifically, we define the Pauli twirling operator
\begin{equation}
    \mathcal T(\Phi)
:=
\mathbb E_{P\sim\mathcal P_n}
\left[
\mathcal U_P^\dagger \circ \Phi \circ \mathcal U_P
\right],
\end{equation}
where $\mathcal U_P(\rho)=P\rho P$ for some superoperator $\Phi$.
The map $\mathcal T$ acts linearly on superoperators and projects onto
their Pauli-diagonal component.
Given a Lindblad generator $\mathcal L$, we define its twirled version
\begin{equation}
    \widetilde{\mathcal L} := \mathcal T(\mathcal L) = \mathbb{E}_{P\sim\mathcal P_n}
\bigl[
\mathcal U_P^\dagger \circ \mathcal L \circ \mathcal U_P
\bigr].
\end{equation}
As shown in Lemma \ref{lem:structure_twirl} below, the twirled generator $\widetilde{\mathcal L}$ is diagonal in the Pauli basis and does not contain a Hamiltonian component, therefore it falls into the class of superoperators analyzed in Section \ref{subsec:pauli_diag}.
Note that we may use $\calT (\calL) = \widetilde \calL = \widetilde \calD$ interchangeably as a result.

\begin{lemma}[Pauli diagonalization under twirling]\label{lem:pauli_diag_twirl}
    Let $\mathcal L$ be a Lindblad generator and let $\widetilde{\mathcal L}=\mathcal T(\mathcal L)$.
    Then the Pauli operators form an eigenbasis of
    $\widetilde{\mathcal L}$.  In particular, if
    \[
    \mathcal L(P)=\sum_{Q\in\mathcal P_n} a_{Q,P} Q ,
    \]
    then
    \[
    \widetilde{\mathcal L}(P)=a_{P,P}P .
    \]
\end{lemma}
\begin{proof}
    Fix Pauli operators $P,Q \in \calP_n$.
    Conjugation by a Pauli maps Pauli operators to themselves up to sign,
    \[
    U P U^\dagger = \chi(U,P) P ,
    \]
    where 
    \[\chi(U,P) := (-1)^{|\{i: U_i \text{ anticommutes with } P_i\}|} \in \{\pm 1\}.\]
    Using the expansion $\calL(P) = \sum_{Q} a_{Q,P} Q$, we have that 
    \begin{align}
        \widetilde{\calL}(P) &= \EE_U[U^\dagger \calL(UPU^\dagger)U] \\
        &= \EE_U[U^{\dagger} \calL(\chi(U,P)P) U]\\
        &= \EE_U [\chi(U,P) U^\dagger \calL(P)U] \\
        &= \sum_{Q} a_{Q,P} \EE_U[\chi(U,P) U^{\dagger} Q U]\\
        &= \sum_Q a_{Q,P} \EE_U[\chi(U,P) \chi(U,Q)]Q 
    \end{align}
    We note that 
    \begin{equation}
        \chi(U,P) \chi(U,Q) = \prod_{i=1}^n \chi(U_i, P_i) \chi(U_i, Q_i).
    \end{equation}
    Since the $U_i$'s are independent and random,
    \[\EE_U[\chi(U,P)\chi(U,Q)] = \prod_{i=1}^n \EE_{U_i}[\chi(U_i, P_i) \chi(U_i, Q_i)].\]
    If $P = Q$, the above expression is trivially $1$.
    If $P \neq Q$, then there is at least one qubit $i$ where $P_i \neq Q_i$.
    On that qubit, the value of the product sign is equally likely to be $+1$ or $-1$. 
    Thus, we have that 
    \begin{equation}
        \EE_U[\chi(U,P) \chi(U,Q)] = \begin{cases}
        1, & P=Q, \\
        0, & P\neq Q.
        \end{cases}
    \end{equation}
    We then arrive that $\calL$ is Pauli-diagonal, i.e.
    \[\widetilde{\calL}(P) = a_{P,P} P.\]
\end{proof}

We now characterize the structure of
the resulting Lindbladian more explicitly. 
In particular, we show that the twirled generator takes the form of a Pauli-diagonal dissipator.

\begin{lemma}[Structure of the twirled Lindbladian]\label{lem:structure_twirl}
Let
\[
\mathcal L(\rho) = -i[H,\rho] + \sum_a \left( L_a\rho L_a^\dagger - \frac12\{L_a^\dagger L_a,\rho\}\right)
\]
be a Lindblad generator, where the jump operators can be expanded as $L_a=\sum_{P\in\mathcal P_n}\gamma_{a,P}P$.
Let
$\widetilde{\mathcal L}=\mathcal T(\mathcal L)$ be its Pauli twirl.
Then
\begin{equation}
    \widetilde{\mathcal L}(\rho)
=
\sum_{P\in\mathcal P_n}
\alpha_P\,(P\rho P-\rho),
\end{equation}
where
$\alpha_P:=\sum_a |\gamma_{a,P}|^2$.
In particular, the twirled dynamics is purely dissipative.
\end{lemma}

\begin{proof}
    Write $\mathcal L = \mathcal L_H + \mathcal D$, where
    \[
    \mathcal L_H(\rho) = -i[H,\rho]
    \]
    and
    \[
    \mathcal D(\rho)
    =
    \sum_a
    \left(
    L_a \rho L_a^\dagger
    -
    \frac12\{L_a^\dagger L_a,\rho\}
    \right).
    \]

    \paragraph{Hamiltonian part.}
    Fix a Pauli operator $P \in \calP_n$.  
    We first show that $\mathcal L_H(P)$ has no component proportional to $P$.
    Indeed,
    \begin{equation}
        \Tr(P^\dagger [H,P]) = \Tr(P^\dagger H P) - \Tr(P^\dagger P H) = \Tr(H) - \Tr(H) = 0,
    \end{equation}
    where we used $P^\dagger P = I$ and cyclicity of trace.
    Thus the coefficient of $P$ in $\mathcal L_H(P)$ is zero.
    By Lemma \ref{lem:pauli_diag_twirl}, Pauli twirling preserves only the diagonal Pauli components, so
    $\mathcal T(\mathcal L_H)=0$.

    \paragraph{Dissipative part.}
    Since the Pauli operators form an orthonormal basis, each jump operator admits a Pauli expansion
    \[
    L_a=\sum_{P\in\mathcal P_n}\gamma_{a,P}P .
    \]
    Substituting this expansion into the dissipator $\calD(\rho)$ gives us 
    \begin{equation}
        \calD(\rho) =\sum_a \left(\sum_{P,Q} \gamma_{a,P}\gamma_{a,Q}^{*} P\rho Q\right) - \frac{1}{2}\left\{\sum_{P,Q} \gamma_{a,P}\gamma_{a,Q}^* QP, \rho\right\},
    \end{equation}
   so $\calD$ is a linear combination of terms invovling $P\rho Q$ and $\{QP,\rho\}.$
    Now, applying the Pauli-frame twirl, 
    \[\calT(\calD(\rho)) := \EE_{U\sim \calP_n}[U^\dagger \calD(U\rho U^\dagger) U].\]
    Firstly, $P \rho Q$ is transformed into the following during the twirling process:
    \[U^\dagger (P(U\rho U^\dagger) Q)U = (U^\dagger PU)\rho (U^\dagger QU).\]
    Conjugation by a Pauli maps Pauli operators to themselves up to sign, so the entire expression becomes 
    \begin{equation}
        (U^\dagger PU)\rho (U^\dagger QU) = \chi(U,P) \chi(U, Q) P\rho Q.
    \end{equation}
    By the same argument as the proof in Lemma \ref{lem:pauli_diag_twirl}, when averaging over random $U \sim \calP_n$, the sign product is only non-zero when $P=Q$, with coefficient $\sum_a \gamma_{a,P}\gamma_{a,P}^*.$
    Thus, twirling eliminates all cross terms $P\rho Q$ with $P \neq Q$ and keeps only $P\rho P.$
    Applying the same argument to the anticommutator term, we obtain
    \begin{equation}
        U^\dagger \{QP, U\rho U^\dagger\} U = \{U^\dagger (QP) U, \rho\} = \{\chi(U, Q)\chi(U, P)QP, \rho\},
    \end{equation}
    where the last equality comes from the fact that $U^\dagger (QP) U = (U^\dagger QU) (U^\dagger PU).$
    Therefore, averaging over random $U$,
    \begin{equation}
        \EE_U [U^\dagger \{QP, U\rho U^\dagger \} U] = \EE_U [\chi(U, P) \chi(U, Q)]\{QP, \rho\}.
    \end{equation}
    The above expression is only non-zero when $P = Q.$
    Thus, $\{QP, \rho\} = \{I, \rho\} = 2\rho.$
    The twirled anticommutator term then only contributes a multiple of $\rho$ with coefficient $\sum_a \gamma_{a,P}\gamma_{a,P}^*$, giving us the final expression 
    \begin{equation}
        \widetilde{\calL}(\rho) = \sum_{P \in \calP_n} \left(\sum_a \gamma_{a,P}\gamma_{a,P}^*\right) (P\rho P - \rho).
    \end{equation}
\end{proof}
To compare the dissipative dynamics before and after twirling, it is convenient to work in the Pauli expansion of the jump
operators.  
Writing $L_a=\sum_{P\in\mathcal P_n}\gamma_{a,P}P$,
we define the coefficient matrix
\begin{equation} \alpha_{P,Q}:=\sum_a\gamma_{a,P}\gamma_{a,Q}^* .
\end{equation}
The matrix $\alpha$ captures how different Pauli components of the
jump operators interact in the dissipator.
The following basic property will be useful.

\begin{lemma}\label{lem:PSD}
    In the Pauli expansion $L_a = \sum_{P \in \calP_n} \gamma_{a, P} P$, the coefficient matrix 
    \[\alpha_{P,Q} := \sum_a \gamma_{a,P} \gamma_{a,Q}^*\]
    is Hermitian positive semidefinite.
\end{lemma}
\begin{proof}
    Consider an arbitrary complex vector $v = (v_P)_{P \in \calP_n},$ we have that 
    \begin{align}
        v^\dagger \alpha v &= \sum_{P,Q} v_P^* \alpha_{P,Q} v_Q \\
        &= \sum_{P,Q} v_P^* \left(\sum_a \gamma_{a,P}\gamma_{a,Q}^*\right) v_Q \\
        &= \sum_a \left(\sum_P v_P^* \gamma_{a,P} \right)
        \left(\sum_Q v_Q^* \gamma_{a,Q} \right)^* \\
        &= \sum_a |\sum_P v_P^* \gamma_{a,P}|^2 \\
        &\geq 0.
    \end{align}
    We also note that $\alpha$ is Hermitian since 
    \begin{equation}
        \alpha_{Q,P} = \sum_a \gamma_{a,Q} \gamma_{a,P}^* = \left(\sum_a \gamma_{a,P} \gamma_{a,Q}^*\right)^* = \alpha_{P,Q}^*.
    \end{equation}
\end{proof}
Furthermore, the locality and bounded degree constraints of the jump operators impose additional structure on the matrix $\alpha$.  
In particular, $k$-local jump operators of bounded degree $\Delta$ imply that $\alpha$ is sparse.
\begin{lemma}\label{lem:alpha_sparsity}
    Let $\{L_a\}_a$ be $k$-local, degree-$\Delta$ jump operators on $n$ qubits.
    Then the coefficient matrix $\alpha$ is sparse in the sense that for every Pauli string $P \in \calP_n$, 
    \begin{equation}
        |\{Q \in \calP_n: \alpha_{P,Q}\neq 0\}| \leq (4\Delta)^k.
    \end{equation}
\end{lemma}

\begin{proof}
    Fix $P \in \calP_n$.
    If $\alpha_{P,Q} \neq 0$, then there must exist some $a$ such that $\gamma_{a, P} \neq 0$ and $\gamma_{a,Q} \neq 0.$
    Since $L_a$ is supported on $S_a$, then any Pauli string with nonzero coefficient in its expansion must also have support contained in $S_a$, so $\supp(P) \subseteq S_a$ and $\supp(Q) \subseteq S_a.$
    Thus, every $Q \in \calP_n$ with $\alpha_{P,Q} \neq 0$ is supported inside $S_a$ that also contains $\supp(P)$, for some index $a$.
    For each qubit $i \in \supp(P)$, there are at most $\Delta$ indices $a$ with $i \in S_a$ by the bounded degree assumption, i.e. each qubit $i$ appears in at most $\Delta$ jump operators.
    Then there are at most $\Delta^{|\supp(P)|}$ jump operators in total whose support contains all the qubits in $\supp(P).$ 
    For each index $a$, the number of Pauli strings $Q$ supported within $S_a$ is at most $4^{|S_a|}.$
    Taking the union over all such $a$ gives us 
    \begin{equation}
        |\{Q: \alpha_{P,Q} \neq 0\}| \leq \Delta^{\supp(P)}4^{\supp(P)} \leq (4\Delta)^k,
    \end{equation}
    where the last inequality follows from the locality assumption.
\end{proof}

We now derive a simple inequality for sparse positive semidefinite
matrices that bounds the total squared magnitude of the entries in
terms of the diagonal entries.
\begin{lemma}\label{lem:sparsity_implies_diagonal}
    Let $\calJ$ be a finite index set. 
    Let $\alpha = (\alpha_{ij})_{i,j \in \calJ}$ be a Hermitian positive semidefinite matrix. 
    Suppose $\alpha$ is $S$-sparse per row in the sense that for every $i \in \calJ$, 
    \[|j \in \calJ: j \neq i \text{ and } \alpha_{ij} \neq 0| \leq S.\]
    Then
    \begin{equation}
        \sum_{i,j\in \calJ} |\alpha_{ij}|^2 \leq (S+1) \sum_{i \in \calJ} |\alpha_{ii}|^2.
    \end{equation}
\end{lemma}
\begin{proof}
    Since $a \succeq 0$, every $2 \times 2$ principal minor is also PSD.
    Thus, for all $i,j$,  
    \[\alpha_{ii} \alpha_{jj} - \alpha_{ij} \alpha_{ji} = \alpha_{ii}\alpha_{jj} - |\alpha_{ij}|^2 \geq 0,\]
    where the first equality comes from the fact that $\alpha$ is Hermitian. 
    This implies that    
    \[|\alpha_{ij}|^2 \leq \alpha_{ii} \alpha_{jj}.\]
    We decompose the sum into diagonal and off-diagonal parts,
    \[\sum_{i,j} |\alpha_{ij}|^2  = \sum_i |\alpha_{ii}|^2 + \sum_{i\neq j} |\alpha_{ij}|^2.\]
    For the off-diagonal part, using $xy \leq \frac{1}{2}(x^2 + y^2)$ for $x,y \geq 0$, we have that 
    \[|\alpha_{ij}|^2 \leq \alpha_{ii}\alpha_{jj} \leq \frac{1}{2} (\alpha_{ii}^2 + \alpha_{jj}^2).\]
    Summing over all the nonzero off-diagonal entries gives us
    \[\sum_{i\neq j} |\alpha_{ij}|^2 \leq \frac{1}{2} \sum_{i \neq j: \alpha_{ij} \neq 0} \alpha_{ii}^2 + \frac{1}{2} \sum_{i \neq j: \alpha_{ij} \neq 0}\alpha_{jj}^2.\]
    Because each row contains at most $S$ nonzero off-diagonal entries, every $\alpha_{ii}^2$ appears at most $S$ times in the first term.
    Therefore, 
    \[\sum_{i\neq j: \alpha_{ij} \neq 0} \alpha_{ii}^2 \leq S \sum_{i} \alpha_{ii}^2.\]
    Since $\alpha$ is Hermitian, $\alpha_{ij} = 0$ implies $\alpha_{ji} = 0$, so the number of nonzeros in a column is equal to the number of nonzeros in a row, which is at most $S$.
    Putting everything together, we obtain 
    \[\sum_{i,j} |\alpha_{ij}|^2 = \sum_i |\alpha_{ii}|^2 + \sum_{i\neq j} |\alpha_{ij}|^2 \leq (S+1) \sum_i |\alpha_{ii}|^2.\]
\end{proof}

We now combine the previous lemmas to obtain a bound on the coefficients $\alpha_{P,Q}$ arising from local jump operators.
This bound will allow us to relate the magnitude of the full coefficient matrix $\alpha$ to its diagonal entries, which correspond precisely to the terms that survive after Pauli twirling.
\begin{corollary}\label{cor:alpha_upper_diag_bound}
Let $\{L_a\}_a$ be $k$-local, degree-$\Delta$ jump operators on $n$
qubits with Pauli expansions
$L_a=\sum_{P\in\mathcal P_n}\gamma_{a,P}P$.
Define the coefficient matrix
$
\alpha_{P,Q}:=\sum_a\gamma_{a,P}\gamma_{a,Q}^*$.
Then
\begin{equation}
    \sum_{P,Q\in\mathcal P_n}|\alpha_{P,Q}|^2 \le \big((4\Delta)^k+1\big)
\sum_{P\in\mathcal P_n}|\alpha_{P,P}|^2 .
\end{equation}
\end{corollary}

\begin{proof}
    By Lemma \ref{lem:PSD}, the matrix $\alpha$ is Hermitian positive semidefinite.
    By Lemma \ref{lem:alpha_sparsity}, it is $S$-sparse per row with $S = (4\Delta)^k.$
    Applying Lemma \ref{lem:sparsity_implies_diagonal} to $\alpha$ therefore tields 
    \[
    \sum_{P,Q}|\alpha_{P,Q}|^2
    \le
    (S+1)\sum_P|\alpha_{P,P}|^2 .
    \]
    Substituting $S=(4\Delta)^k$ completes the proof.
\end{proof}

Next, we relate the magnitude of the original dissipator $\calD$ to that of the twirled dissipator $\widetilde \calD$.
To do so, we first express superoperators in the Pauli basis and relate their coefficients to the normalized Frobenius norm.

\begin{lemma}\label{lem:isometry_exp}
Let $d=2^n.$ 
For $P,Q \in \calP_n$, define the linear map $\calE_{P,Q}:= \mathsf{L}(\C^d) \to \mathsf{L}(\C^d)$ as 
\[\calE_{P,Q}(X) = PXQ.\]
Then the maps $\{\mathcal E_{P,Q}\}_{P,Q\in\mathcal P_n}$ are orthogonal with respect to the Hilbert--Schmidt inner product on superoperators, in the sense that
\begin{equation}
    \Tr(\calE^\dagger_{P,Q}\calE_{P',Q'}) = d^2 \delta_{P,P'} \delta_{Q,Q'}.
\end{equation}
Consequently, for any coefficients $\{\alpha_{P,Q}\} \subset \C$, the map $\calE(X) = \sum_{P,Q \in \calP_n} \alpha_{P,Q} PXQ$ satisfies 
\begin{equation}
    \|\calE\|_F^2 := \frac{1}{d^2} \Tr(\calE^\dagger \calE) = \sum_{P,Q \in \calP_n} |\alpha_{P,Q}|^2.
\end{equation}
\end{lemma}

\begin{proof}
    By definition of the Hilbert--Schmidt inner product on superoperators,
    \begin{equation}
        \Tr(\calE_{P,Q}^\dagger \calE_{P',Q'}) = \sum_{i,j=1}^d \Tr(\ket{j}\bra{i}\calE_{P,Q}^\dagger (\calE_{P',Q'} (\ket{i}\bra{j}))).
    \end{equation}
    Since $\calE_{P',Q'}(\ket{i}\bra{j}) = P'\ket{i}\bra{j}Q'$ and $\calE^\dagger_{P,Q} (X) = P^\dagger X Q^\dagger$, we obtain 
    \begin{equation}
        \Tr(\calE_{P,Q}^\dagger \calE_{P',Q'}) = \sum_{i,j} \Tr(\ket{j}\bra{i} P^\dagger P'\ket{i}\bra{j})Q'Q^\dagger) = \sum_{i,j} \bra{i}P^\dagger P' \ket{i} \bra{j} Q'Q^\dagger \ket{j} = \Tr(P^\dagger P') \Tr(Q^\dagger Q)
    \end{equation}
    Note that any Pauli string is a tensor product $P = \bigotimes_{i=1}^n P_i$ with $P_i \in \calP.$
    Then, 
    \[\Tr(P^\dagger P') = \Tr\left(\bigotimes_{i=1}^n P_i^\dagger P_i'\right) = \prod_{i=1}^{n} \Tr(P_i^\dagger P_i').\]
    For single-qubit Paulis, note that $\Tr(I^\dagger I) = \Tr(I) = 2, \Tr(X) = \Tr(Y) = \Tr(Z) = 0$, and $\Tr(P_i^\dagger P_i') = 0$ when $P_i \neq P_i'.$
    This gives us $\Tr(P^\dagger P') = d \delta_{P, P'}$, and we obtain the desired expression 
    \[\Tr(\calE^\dagger_{P,Q}\calE_{P',Q'}) = d^2 \delta_{P,P'} \delta_{Q,Q'}.\]
    Lastly, we have that 
    \begin{equation}
        \calE^\dagger \calE = \left(\sum_{P,Q} \alpha_{P,Q}^* \calE_{P,Q}^\dagger\right) \left(\sum_{P',Q'} \alpha_{P',Q'} \calE_{P',Q'}\right) = \sum_{P,Q}\sum_{P',Q'} \alpha_{P,Q}^* \alpha_{P',Q'} \calE_{P,Q}^\dagger \calE_{P',Q'}. 
    \end{equation}
    Finally, by the definition of the normalized Frobenius norm, 
    \begin{equation}
        \|\calE\|_F^2 = \frac{1}{d^2} \sum_{P,Q}\sum_{P',Q'} \alpha_{P,Q}^* \alpha_{P',Q'} \Tr(\calE_{P,Q}^\dagger \calE_{P',Q'}) = \sum_{P,Q} \alpha_{P,Q}^* \alpha_{P,Q} = \sum_{P,Q}|\alpha_{P,Q}|^2.
    \end{equation}
\end{proof}
We now relate the normalized Frobenius norm of the twirled dissipator to the diagonal coefficients of $\alpha_{P,P}$.

\begin{lemma}\label{lem:twirled_lower_bound}
    Consider the twirled dissipator 
    \[\widetilde{\calD}(\rho) := \calT(\calL) = \sum_{P \in \calP_n} \alpha_{P,P} (P\rho P - \rho).\]
    Then 
    \begin{equation}
        \sum_{P \neq I} |\alpha_{P,P}|^2 \leq \|\widetilde{\calD}\|_F^2
    \end{equation}
\end{lemma}
\begin{proof}
    Let $R \in \calP_n$ be a Pauli operator.
    Since Pauli operators either commute or anticommute, we have that $PRP = \chi(P,R)R$ where $\chi(P,R) \in \{\pm 1\}.$
    Then, 
    \begin{equation}
        \widetilde{\calD}(R) = \sum_P \alpha_{P,P} (\chi(P,R)R - R) = \left(\sum_P \alpha_{P,P} (\chi(P,R) -1)\right) R
    \end{equation}
    Then, each Pauli $R$ is an eigenoperator of $\widetilde{D}$, with eigenvalue 
    $\lambda_R := \sum_P \alpha_{P,P} (\chi(P,R) - 1).$
    We can then write the normalized Frobenius norm of $\widetilde{\calD}$ as 
    \begin{align}
        \|\widetilde{\calD}\|_F^2 &= \frac{1}{d^2} \sum_{R \in \calP_n} |\lambda_R|^2\\
        &= \frac{1}{d^2}\sum_R \left|\sum_P \alpha_{P,P} (\chi(P,R) - 1)\right|^2 \\
        &= \frac{1}{d^2} \sum_{P,Q} \alpha_{P,P} \alpha_{Q,Q}^* \sum_R (\chi(P,R) - 1) (\chi(Q,R) - 1).
    \end{align}

    For $P\neq I$, exactly half of the Paulis commute with $P$ and half anticommute, so $\sum_R \chi(P,R) = 0$.
    Moreoever, combining the facts that $\chi(P,R) \chi(Q,R) = \chi(PQ,R)$ and $\sum_R 1 = 4^n = d^2$, we get $\sum_{R} \chi(P,R) \chi(Q,R) = d^2 \delta_{P,Q}$.
    Using these identities, we obtain for $P,Q \neq I$, 
    \begin{align}
        \sum_R (\chi(P,R) - 1) (\chi(Q,R) - 1) &= 
    \begin{cases}
        2d^2, & P = Q \neq I,\\
        d^2, & P \neq Q
    \end{cases} \\
    &= d^2 (1+\delta_{P,Q})
    \end{align}
    Combing everything together, 
    \begin{align*}
        \|\widetilde{\calD}\|_F^2 &= \frac{1}{d^2} \sum_{P,Q \neq I} \alpha_{P,P} \alpha_{Q,Q}^* d^2 (1+\delta_{P,Q}) \\
        &= \left|\sum_{P\neq I} \alpha_{P,P}\right|^2 + \sum_{P\neq I}|\alpha_{P,P}|^2 \geq \sum_{P\neq I}|\alpha_{P,P}|^2. 
    \end{align*}
\end{proof}

Combining the previous bounds with the sparsity properties of the coefficient matrix $\alpha$, we are now ready to relate the normalized Frobenius norms of the original dissipator and its Pauli-twirled version.

\begin{lemma}\label{lem:coeff_dissipator_norm_bound}
    Let 
    \[\calD(\rho) = \sum_a (L_a\rho L_a^\dagger  - \frac{1}{2} \{L_a^\dagger L_a, \rho \}),\]
    where $\{L_a\}_a$ are $k$-local, degree-$\Delta$ jump operators with Pauli expansion $L_a = \sum_{P \in \calP_n} \gamma_{a,P}P.$
    Define the coefficients $\alpha_{P,Q}:= \sum_{a}\gamma_{a,P} \gamma_{a,Q}^*.$ 
    Then the Pauli-twirled dissipator 
    \[\widetilde{\calD}(\rho) := \calT(\calD(\rho)) = \sum_{P\in\calP_n} \alpha_{P,P} (P\rho P - \rho)\] satisfies
    \[\|\calD\|_F \leq 2((4\Delta)^k+1) \|\widetilde{\calD}\|_F.\]
\end{lemma}

\begin{proof}
    We decompose the dissipator as $\calD = \calD_1 - \frac{1}{2}\calD_2 - \frac{1}{2}\calD_3,$ where 
    \[\calD_1 (\rho) :=  \sum_a L_a \rho L_a^\dagger, \quad \calD_2(\rho) := \sum_a (L_a^\dagger L_a) \rho, \quad \calD_3(\rho):= \sum_a \rho (L_a^\dagger L_a).\]

    \paragraph{Bounding $\calD_1$.}
    We may write
    \[\calD_1 (\rho) = \sum_a \sum_{P,Q} \gamma_{a,P}\gamma_{a,Q}^* P\rho Q = \sum_{P,Q}\alpha_{P,Q}P\rho Q.\]
    Applying Lemma \ref{lem:isometry_exp} directly with $\calE_{P,Q}(X) = PXQ$ yields 
    \begin{equation}
        \|\calD_1\|_F^2 = \sum_{P,Q} |\alpha_{P,Q}|^2.
    \end{equation}

    \paragraph{Bounding $D_2$ and $D_3$.}
    For each $a$, we compute
    \begin{equation}
        L_a^\dagger L_a = \left(\sum_P\gamma_{a,P}P\right)^\dagger \left(\sum_Q\gamma_{a,Q}Q\right) = \sum_{P,Q}\gamma_{a,P}^*\gamma_{a,Q} PQ.
    \end{equation}
    Summing over $a$ and using $\alpha_{Q,P}:= \sum_a \gamma_{a,P}^* \gamma_{a,Q}$, 
    gives
    \begin{equation}
        \sum_{a} L_a^\dagger L_a = \sum_{P,Q} \alpha_{Q,P} PQ = \sum_{R \in \calP_n} b_R R,
    \end{equation}
    where we denote $R:= PQ$ and $b_R := \sum_{Q\in \calP_n} \alpha_{Q, RQ}.$
    Therefore, we can write 
    \begin{equation}
        [\calD_2(\rho) = \sum_a (L_a^\dagger L_a)\rho = \sum_R b_R R\rho, \quad \calD_3(\rho) = \sum_R b_R \rho R
    \end{equation}
    Applying Lemma \ref{lem:isometry_exp} to the maps $\rho \to R\rho I$ and $\rho \to I\rho R$ gives us the bounds 
    \begin{equation}
        \|\calD_2\|_F^2 = \|\calD_3\|_F^2 = \sum_R |b_R|^2.
    \end{equation}
    By Lemma \ref{lem:alpha_sparsity} and the fact that $\alpha$ is Hermitian, for a fixed $R$, there are at most $S = (4\Delta)^k$ values of $Q$ that can contribute as nonzero terms to $b_R.$
    Therefore, by Cauchy-Schwarz applied to a sum of at most $S$ terms, we have that 
    \begin{equation}
        |b_R|^2 = \left|\sum_Q \alpha_{Q, RQ}\right|^2 \leq S\sum_Q \left|\alpha_{Q,RQ}\right|^2.
    \end{equation}
    Summing over $R$ and changing variables gives
    \begin{equation}
        \|\calD_2\|_F^2 = \|\calD_3\|_{F}^2 =\sum_R |b_R|^2 \leq S \sum_{P,Q}|\alpha_{P,Q}|^2 = \|\calD_1\|_F^2.
    \end{equation}

    \paragraph{Combining the bounds.}
    By triangle inequality, 
    \begin{equation}
        \|\calD\|_F \leq \|\calD_1\|_F + \frac{1}{2}\|\calD_2\|_F + \frac{1}{2}\|\calD_3\|_F \leq (1+\sqrt{S})\|\calD_1\|_F.
    \end{equation}
    By Corollary \ref{cor:alpha_upper_diag_bound}, 
    \begin{equation}
        \|\calD_1\|_F^2 \leq (S+1) \sum_P |\alpha_{P,P}|^2 \leq (S+1)(|\alpha_{I,I}|^2 + \|\widetilde{\calD}\|_F^2) = (S+1)\|\widetilde{\calD}\|_F^2 ,
    \end{equation}
    where the second inequality follows from Lemma \ref{lem:twirled_lower_bound} and the last equality comes from the fact that each jump operator $L_a$ is traceless.  \color{black} 
    Putting everything together and using the fact that $(1+\sqrt{S})\sqrt{S+1} \leq 2(S+1)$ for $S \geq 0$,  we obtain the final bound 
    \begin{equation}
        \|\calD\|_F \leq (1+\sqrt{S}) (S + 1)^{1/2} \|\widetilde{\calD}\|_F \leq 2(S+1) \|\widetilde{\calD}\|_F = 2((4\Delta)^{k}+1)\|\widetilde{\calD}\|_F.
    \end{equation}
\end{proof}

\subsection{Implementation of the twirled evolution}\label{subsec:trotter}
In the analysis of the previous section we considered the dynamics generated by the twirled Lindbladian $\mathcal T(\mathcal L)$.
However, experimentally we only have black-box access to the evolution channel
\[
\mathcal E_t = e^{t\mathcal L}.
\]
Since in general
\[
\mathcal T(e^{t\mathcal L}) \neq e^{t\mathcal T(\mathcal L)},
\]
we cannot directly implement the evolution generated by $\mathcal T(\mathcal L)$.
In this section we show that the twirled dynamics can nevertheless be simulated by repeatedly applying short-time evolutions of $\mathcal T(e^{t\mathcal L})$.
Recall that the Pauli twirling operator is defined as
\[
\mathcal T(\Phi) :=
\mathbb E_{P\sim\mathcal P_n}
\left[ \calU_P^\dagger \circ \Phi \circ \calU_P \right].
\]
We now define the Pauli-twirled channel of $\calE_t$ as
\[
\widetilde{\mathcal E}_t
:= \calT(\calE_t) = 
\mathbb E_{P\sim\mathcal P_n}
\big[\calU_P^\dagger \circ \mathcal E_t \circ \calU_P\big].
\]
We first bound the difference between the channel-level twirl $\mathcal T(e^{t\mathcal L})$ and the evolution $e^{t\mathcal T(\mathcal L)}$ generated by the twirled Lindbladian.

\begin{lemma}\label{lem:gen_twirl_diff}
    Let $\mathcal L$ be a time-independent Lindblad generator and let $\mathcal T$ denote the Pauli twirling operator on superoperators with 
    \[
    \mathcal T(\Phi):=
    \mathbb E_{P\in\mathcal P_n}
    \bigl[
    \mathcal U_P^\dagger \circ \Phi \circ \mathcal U_P
    \bigr].
    \]
    For $t \geq 0$, 
    \begin{equation}
        \|\calT(e^{t\calL}) - e^{t \calT(L)}\|_{\diamond} \leq \frac{t^2}{2}\|\calT(\calL^2) - (\calT(\calL))^2\|_{\diamond} + \frac{t^3}{3}\|\calL\|_{\diamond}^3.
    \end{equation}
    % \YT{This bound is too loose. We can get rid of the $e^{t\|\calL\|_F}$ factor using the CPTP nature of the dynamics.}
\end{lemma}
\begin{proof}
    We begin by expanding the two exponentials:
    \[\calT(e^{t\calL}) = \sum_{k \geq 0} \frac{t^k}{k!} \calT(\calL^k) = I + t \calT(\calL) + \frac{t^2}{2} \calT(\calL^2) + \sum_{k\geq 3}  \frac{t^k}{k!} \calT(\calL^k).\]
    Similarly, 
    \[e^{t\calT(\calL)} = I + t \calT(\calL) + \frac{t^2}{2} (\calT(\calL))^2 + \sum_{k\geq 3}  \frac{t^k}{k!} (\calT(\calL))^k.\]
    The expressions agree up to second order, so
    \begin{equation}
        \calT(e^{t\calL}) - e^{t\calT(\calL)} \leq \frac{t^2}{2}(\calT(\calL^2) - (\calT(\calL))^2) + \sum_{k\geq 3}  \frac{t^k}{k!}(\calT(\calL^k) - \calT(\calL)^k).
    \end{equation}
    Applying the triangle inequality and the diamond norm, 
    \begin{equation}
        \|\calT(e^{t\calL}) - e^{t\calT(\calL)}\|_{\diamond} \leq \frac{t^2}{2}\|\calT(\calL^2) - \calT(\calL)^2\|_{\diamond} + \left\|\sum_{k\geq 3}  \frac{t^k}{k!}\left(\calT(\calL^k)-(\calT(\calL))^k \right)\right\|_{\diamond}.
    \end{equation}
    We can use the integral form of the Taylor remainder for the exponential and rewrite 
    \begin{equation}
        \sum_{k \geq 3}\frac{t^k}{k!} \calL^k = \int_{0}^t \frac{(t-s)^2}{2} \calL^3 e^{s\calL} ds,
    \end{equation}
    so by linearity of $\calT$,
    \begin{equation}
        \sum_{k \geq 3}\frac{t^k}{k!} \calT(\calL^k) = \calT\left( \int_{0}^t \frac{(t-s)^2}{2} \calL^3 e^{s\calL} ds\right),
    \end{equation}
    and similarly, 
    \begin{equation}
        \sum_{k \geq 3}\frac{t^k}{k!} \calT(\calL)^k = \int_{0}^t \frac{(t-s)^2}{2} \calT(\calL)^3 e^{s\calT(\calL)}ds.
    \end{equation}
    We can then bound the first tail as
    \begin{align}
        \left\|\sum_{k\geq 3} \frac{t^k}{k!}\calT(\calL^k) \right\|_{\diamond} &\leq \int_{0}^t \frac{(t-s)^2}{2} \| \calT(\calL^3 e^{s\calL})\|_{\diamond} ds\\
        &\leq \int_{0}^t \frac{(t-s)^2}{2} \|\calL^3\|_{\diamond} \|e^{s\calL}\|_{\diamond} ds \\
        &\leq \frac{t^3}{6} \|\calL\|_{\diamond}^3,
    \end{align}
    where the second inequality comes from the fact that $\calT$ is a convex combination of unitary conjugations, and the final inequality uses that $e^{s\calL}$ is CPTP for $s\geq 0$, so $\|e^{s\calL}\|_{\diamond} = 1$, and the fact that diamond norm is submultiplicative.
    Likewise 
    \begin{equation}
         \left\|\sum_{k\geq 3} \frac{t^k}{k!}\calT(\calL)^k \right\|_{\diamond} \leq \int_{0}^t \frac{(t-s)^2}{2} \|\calT(\calL)\|_{\diamond}^3 \|e^{s\calT(\calL)}\|_{\diamond} ds \leq \frac{t^3}{6} \|\calL\|_{\diamond}.
    \end{equation}
    Combining the above bounds and applying triangle inequality, 
    we obtain
    \begin{equation}
        \|\calT(e^{t\calL}) - e^{t \calT(L)}\|_{\diamond} \leq \frac{t^2}{2}\|\calT(\calL^2) - (\calT(\calL))^2\|_{\diamond} + \frac{t^3}{3}\|\calL\|_{\diamond}^3,
    \end{equation}
    as claimed.

\end{proof}

We now extend the above approximation to longer evolution times by repeatedly applying the short-time channel $\mathcal T(e^{t\mathcal L})$.
The following lemma quantifies the error incurred when composing $m$ such steps.
\begin{lemma}\label{lem:m_step_diff_norm}
    Let $\mathcal{T}$ denote the Pauli twirling operator and let $\mathcal{L}$ be a Lindblad generator. 
    Let $t>0$ and $m\in\mathbb{N}$, and define the total evolution time $T := t\cdot m$.
    Then
    \begin{equation}
        \|(\calT(e^{t\calL}))^m - e^{tm\calT(\calL)}\|_{\diamond} \leq m \left( \frac{t^2}{2}\|\calT(\calL^2) - (\calT(\calL))^2\|_{\diamond} + \frac{t^3}{3}\|\calL\|_{\diamond}^3 \right).
    \end{equation}
\end{lemma}
% \YT{Here we should use the diamond norm rather than Frobenius norm. The exponential factor can be removed using the fact that a channel has diamond norm 1.}
\begin{proof}
    Let $A:= \calT(e^{t\calL})$ and $B:= e^{t\calT(\calL)}$.
    Expanding the difference as telescoping sum, we have 
    \[A^m - B^m = \sum_{j=0}^m A^{m-1-j}(A-B)B^j.\]
    Applying the triangle inequality and submultiplicativity of diamond norm gives
    \begin{equation}
        \|A^m - B^m\|_{\diamond} \leq \sum_{j=0}^m \|A\|_{\diamond}^{m-1-j} \|A-B\|_{\diamond} \|B\|_{\diamond}^j.
    \end{equation}
    Since $e^{t\calL}$ is CPTP and $\calT$ is a convex combination of unitary conjugations, $\|A\|_{\diamond} = \|B\|_{\diamond} = 1.$
    Therefore, 
    \begin{equation}
        \|A^m - B^m\|_{\diamond} \leq \sum_{j=0}^{m-1} \|A-B\|_{\diamond} = m \|A - B\|_{\diamond} \leq m \left( \frac{t^2}{2}\|\calT(\calL^2) - (\calT(\calL))^2\|_{\diamond} + \frac{t^3}{3}\|\calL\|_{\diamond}^3 \right),
    \end{equation}
    where the last inequality comes from a direct application of Lemma \ref{lem:gen_twirl_diff}.

\end{proof}

\subsection{Detection algorithm}\label{subsec:certification}
Using the implementation of the twirled evolution from the previous section, we now present a procedure for detecting the absence of dissipation in a Lindbladian dynamics.\\\\
\noindent
\setlength{\fboxsep}{7pt}
\fbox{
\begin{minipage}{\dimexpr\linewidth-2\fboxsep-2\fboxrule\relax}
\paragraph{Algorithm description.}
\mbox{}\\[0.3em]
\textit{Inputs:}
\begin{itemize}
\item Oracle access to the channel $\mathcal E_t = e^{t\mathcal L}$ for any $t \ge 0$.
\item Precision parameter $\epsilon > 0$ and failure probability $\delta \in (0,1)$.
\end{itemize}
\textit{Promises:}
\begin{itemize}
\item $\mathcal L(\rho) = -i[H,\rho] + \mathcal D(\rho)$ is a time-independent
Lindblad generator on $n$ qubits with $\|\mathcal L\|_{\diamond} \le L$.
\item Jump operators are $k$-local with degree $\Delta$.
\end{itemize}
\textit{Output:} \textsf{ACCEPT} or \textsf{REJECT}.
\begin{enumerate}
    \item Repeat for $i \in 1, \ldots, R$ times
    \begin{enumerate}
        \item Sample $t$ uniformly from the interval $t \sim U(\left[0, \frac{2}{\epsilon'}\right]),$
        where $\epsilon' = \frac{\epsilon}{2((4\Delta)^k+1)}.$ 
        Set $\tau := t/m.$
        \item Prepare a Bell pair, $\ket{\Phi}_{\sA\sB}$, which is a maximally entangled state on $n+n$ qubits.
        \item Repeat for $j \in 1, \ldots, m$ times
        \begin{enumerate}
            \item Sample $P_{i,j} \gets \{I,X,Y,Z\}^{\otimes n}$ uniformly random.
            \item On register $\sA$, apply $P_{i,j}$, followed by the black-box channel $e^{\tau\calL}$ and $P_{i,j}^\dagger$, sequentially.
        \end{enumerate}
        \item Measures registers $\sA, \sB$ in the $n$-qubit Bell basis.
        \item Store 
        \[X_i = \begin{cases}
            0, & \text{if the output state is } \ket{\Phi}_{\sA\sB},\\
            1, & \text{otherwise}.
            \end{cases}\]
        % \item Output \textsf{REJECT} if the outcome remains $\ket{\Phi}_{\sA\sB}.$
        % \item Let $X_i \gets 1$ if the outcome is the maximally mixed state $\ket{\Phi}$, and $X_i \gets 0$ otherwise.
        % \item Output \textsf{REJECT} if $X_i = 0.$
    \end{enumerate}
    \item Compute $\hat{q} := \frac{1}{R} \sum_{i=1}^R X_i$.
    \item Output \textsf{``Hamiltonian Only"} if $\hat{q} \leq \theta := \frac{9^{-(k-1)}}{120}$, otherwise output \textsf{``Dissipation detected"}.
\end{enumerate}

\end{minipage}
}
\\\\

We now analyze the performance of the above detection procedure.
Combining the structural properties of the twirled Lindbladian from
Section \ref{subsec:twirling} with the implementation guarantees from Section \ref{subsec:trotter} yields the following result.

\begin{theorem}[Certification of Dissipation]\label{thm:final_certification}
Let $\mathcal L(\rho) = -i[H,\rho] + \mathcal D(\rho)$
be an $n$-qubit Lindbladian where the jump operators $\{L_a\}_a$ are
$k$-local and of degree $\Delta$ and $\|\calL\|_{\diamond} \leq L$. Suppose we are given black-box access
to the evolution channel $e^{t\mathcal L}$.
Then there exists a randomized procedure that
distinguishes the cases
\[
\mathcal D = 0
\qquad\text{and}\qquad
\|\mathcal D\|_F \ge \epsilon
\]
with success probability at least $1-\delta$.
The procedure uses
\[
% R = \left\lceil \frac{40}{3} \cdot 9^k\log\frac{1}{\delta} \right\rceil
 R = \left\lceil 3200 \cdot 9^{2k-1} \cdot \log\left(\frac{2}{\delta}\right)\right\rceil
\]
samples, 
% \[Q =\left\lceil\left(1 +\frac{e}{3}\right) 2560\cdot 9^{2k-1} \left(((4\Delta)^k + 1)L\frac{\log(1/\delta)}{\epsilon}\right)^2 \exp\left(\frac{2((4\Delta)^k+1)L}{\epsilon}\right)\right\rceil\]
\[Q = \left\lceil 12288000 \cdot 9^{3k-2}((4\Delta)^k+1)^2 L^2\frac{\log(2/\delta)}{\epsilon^2}\right\rceil,\]
number of oracle queries to the channel $e^{t\calL}$, and total evolution time at most
\[
% T = \frac{80}{3} \cdot 9^k((4\Delta)^k +1) \cdot \frac{\log(1/\delta)}{\epsilon}.
 T = \left\lceil 12800 \cdot 9^{k-1} ((4\Delta)^k + 1) \cdot \frac{\log(2/\delta)}{\epsilon}\right\rceil.
\]
\end{theorem}

\begin{proof}
    Fix one repetition of the algorithm, and condition on the sampled time $t \in [0,t_{\max}]$, where $t_{\max}$ is the largest time sampled per iteration of the algorithm.
    Let $\tau = t/m$ and $\widetilde{\calL} := \calT(\calL)$.
    We first account for the implementation error introduced by replacing the ideal generator-twirled evolution $e^{t\calL}$ with the short-time randomized product $(\calT(e^{t\calL}))^m$.
    By Lemma \ref{lem:diff_bell_identity_prob} and Lemma \ref{lem:m_step_diff_norm}:
    \begin{align}
    \left|I\left((\calT(e^{\tau\calL}))^m\right) - I(e^{\tau m\widetilde{\calL}})\right| &\leq \frac{1}{2} \left\|(\calT(e^{\tau\calL}))^m - e^{\tau m\widetilde{\calL}}\right\|_{\diamond}\\
        &\leq   m \left( \frac{\tau^2}{4}\|\calT(\calL^2) - (\calT(\calL))^2\|_{\diamond} + \frac{\tau^3}{6}\|\calL\|_{\diamond}^3 \right)\\
        &\leq m\left( \frac{\tau^2 L^2}{2} + \frac{\tau^3}{6}L^3\right)\\
        &\leq \frac{(t_{\max})^2L^2}{2m} + \frac{(t_{\max})^3 L^3}{6m^2}\\
        &:= \epsilon_{\text{trott}}.
    \end{align}
    where the second-to-last inequality follows from the assumption $\|\calL\|_{\diamond} \leq L$ and from contractivity as well as submultiplicativity of $\calT$, 
    \[\|\calT(\calL^2) - (\calT(\calL))^2\|_{\diamond} \leq \|\calT(\calL^2)\|_{\diamond} + \|\calT(\calL)^2\|_{\diamond} \leq 2L^2,\]
    and the last inequality comes from $t \leq T$.
    Equivalently, since $X_i = 1$ denotes a non-identity Bell outcome,
    \begin{equation}\label{eqn:imp_error}
        \left|\Pr\left[X_i = 1|t\right] - \left(1 - e^{t\widetilde{\calL}}\right)\right| \leq \epsilon_{\text{trott}}
    \end{equation}
    
    We first consider the Hamiltonian case where $\calD=0$, where $\calL = i[H, \cdot]$.
    Since Pauli twirling removes Hamiltonian commutators at the generator level,
    \[\widetilde{\calL} := \calT(\calL) = 0.\]
    % therefore 
    % \[\calT(e^{\tau\calL}) = \EE_{P\sim\calP_n}[\calU_{P}^\dagger \circ e^{\tau \calL} \circ \calU_{P}]\]
    % is still a unitary conjugation.
    % Therefore, $\calT(e^{\tau\calL})$ is a convex combination of unitary conjugations, and in particular it preserves the maximally mixed state and has superoperator trace $d^2$.
    % Equivalently, 
    % \begin{equation}
    %     I(\calT(e^{\tau \calL})) = \frac{1}{d^2}\Tr(\calT(e^{\tau \calL})) = 1.
    % \end{equation}
    % Since the Bell identity test is multiplicative under sequential application of channels on one half of $|\Phi\rangle$, each slice preserves the identity Bell outcome with probability $1$.
    Therefore, $e^{t\widetilde{\calL}} = I$ for every $t$, and $I(e^{t\widetilde{\calL}}) = 1$.
    By Equation (\ref{eqn:imp_error}), we then have for every sampled $t \in [0,T]$,
    \begin{equation}
        \left|\Pr\left[X_i = 1|t\right]\right| = \EE_t[X_i] \leq \epsilon_{\text{trott}}.
    \end{equation}
    We choose $m$ such that $\epsilon_{\text{trott}} \leq \theta/2$, 
    this then gives us 
    \begin{equation}
        \EE[\hat{q}] = \frac{1}{R}\sum_i^R\EE_t[X_i] \leq \frac{\theta}{2}.
    \end{equation}
    The algorithm incorrectly outputs \textsf{``Dissipation detected"} only if $\hat{q} > \theta$, so by Hoeffding’s inequality, the failure probability when $\calD = 0$ is 
    \begin{equation}
        \Pr\left[\hat{q} > \theta |\calD = 0\right] \leq \Pr\left[\hat{q} - \EE[\hat{q}] \geq \frac{\theta}{2}\right] \leq \exp \left(-\frac{R \theta^2}{2}\right) \leq \frac{\delta}{2},
    \end{equation}
    where the last inequality comes from the choice of $R \geq \frac{2}{\theta^2} \log \left(\frac{2}{\delta}\right).$

    \medskip
    We now consider the case $\|\mathcal D\|_F \ge \epsilon$.
    Since Pauli twirling removes the Hamiltonian component,the twirled generator is purely dissipative, i.e. $\widetilde{\mathcal L} := \mathcal T(\mathcal L) = \mathcal T(\mathcal D).$
    By Lemma \ref{lem:coeff_dissipator_norm_bound}, the twirled dissipator $\widetilde{\mathcal D}:=\calT(\calD)=\widetilde{\calL}$ satisfies
    \begin{equation}
         \|\mathcal D\|_F \le 2((4\Delta)^k+1)\|\widetilde{\mathcal D}\|_F .
    \end{equation}
    This then gives us 
    \begin{equation}
        \|\widetilde{\calL}\|_F \geq \epsilon' := \frac{\epsilon}{2((4\Delta)^k + 1)}.
    \end{equation}
    We can directly apply Lemma \ref{lem:pauli-diag} to the now diagonal generator $\widetilde{\calL}$ and obtain that with the choice of $r = 1/2$, 
    \begin{equation}
         \Lambda(\widetilde{\mathcal L}, r\|\widetilde{\mathcal L}\|_F)
        \ge \frac{9^{-(k-1)}}{16} := c_0 .
    \end{equation}
    By Lemma \ref{lem:pauli-diag}, with $t$ drawn uniformly from $[0, 2/\epsilon']$ we obtain
    \begin{equation}
        \Pr_t\!\left[I(e^{t\widetilde{\mathcal L}})\le 1-\frac{2c_0}{3}\right] \ge \frac{2}{5}.
    \end{equation}
    Then, the expected ideal non-identity Bell probability satisfies
    \begin{equation}
        \EE_t[1 - I(e^{t\widetilde{\calL}})] \geq \frac{2}{5} \cdot \frac{2c_0}{3} = \frac{4c_0}{15} = \frac{9^{-(k-1)}}{60}:= p_0.
    \end{equation}
    Again, from the implementation-error bound, we have that 
    \begin{equation}
        \EE_t[X_i] \geq p_0 - \epsilon_{\text{trott}} \geq \frac{3p_0}{4},
    \end{equation}
    where the last inequality comes from the choice of $m$ such that $\epsilon_{\text{trott}} \leq \frac{p_0}{4}$.
    Choosing the threshold that $\theta := \frac{p_0}{2}$, we have that 
    \begin{equation}
        \EE[\hat{q}] = \frac{1}{R}\sum_{i=1}^R \EE_t[X_i] \geq \frac{3\theta}{2}.
    \end{equation}
    We can then compute that the algorithm incorrectly outputs ``\textsf{Hamiltonian only}" when $\|\calD\|_F \geq \epsilon$ with probability 
    \begin{equation}
        \Pr[\hat{q} \leq \theta] \leq \Pr\left[\hat{q} - \EE[\hat{q}] \leq -\frac{\theta}{2}\right] \leq \exp\left(-\frac{R\theta^2}{2}\right) \leq \frac{\delta}{2}.
    \end{equation}
    We can then bound the overall failure probability as 
    \begin{equation}
        \Pr[\text{algorithm is incorrect}] \leq \Pr\left[\hat{q} > \theta | \calD = 0\right] + \Pr\left[\hat{q} \leq \theta | \|\calD\|_F \geq \epsilon\right] \leq \delta.
    \end{equation}
    Note that the choice of number of rounds $R$ needs to satisfy 
    \begin{equation}
        R = \left\lceil 3200 \cdot 9^{2k-1} \cdot \log\left(\frac{2}{\delta}\right)\right\rceil
    \end{equation}
    for the failure probability to hold, and since each round requires evolution time $t_{\max} = \frac{2}{\epsilon'}$, we have that the total time evolution needed is at most
    \begin{equation}
        T = t_{\max} \cdot R = \left\lceil 12800 \cdot 9^{k-1} ((4\Delta)^k + 1) \cdot \frac{\log(2/\delta)}{\epsilon}\right\rceil.
    \end{equation}
    Lastly, the choice of the number of trotterization steps $m$ that satisfies $\epsilon_{\text{trott}}\leq p_0/4$
    is 
    \begin{equation}
        m = \left\lceil 3840 \cdot 9^{k-1} \frac{((4\Delta)^k+1)^2 L^2}{\epsilon^2}\right\rceil ,
    \end{equation}
    and this gives us the overall query complexity 
    \begin{equation}
        Q = m \cdot R = \left\lceil 12288000 \cdot 9^{3k-2}((4\Delta)^k+1)^2 L^2\frac{\log(2/\delta)}{\epsilon^2}\right\rceil.
    \end{equation}
    This completes the proof.

\end{proof}

\paragraph{Acknowledgment.}
We are especially grateful for Yu Tong and Yongtao Zhan for their guidance throughout the development of this work, and Matthias Caro for proposing the framework of testing dissipation in the open systems as well as his helpful feedback on this manuscript.
We thank Ewin Tang for pointing out the connection to purity testing, as well as Laura Lewis and Savar Sinha for insightful discussions. 
This work is supported by a Gates Cambridge scholarship.

\bibliographystyle{unsrt}
\bibliography{ref}
\end{document}